\documentclass[a4paper]{article}

\usepackage{fullpage}

\usepackage{color}
\usepackage{graphicx}
\usepackage{amsmath}
\usepackage{amssymb}
\usepackage{amsthm}
\usepackage[noend]{algpseudocode}
\usepackage{algorithm}
\usepackage{tikz}
\usepackage{enumitem}
\usepackage{booktabs}
\usepackage{tabularx}
\usepackage{array}
\usepackage{hyperref}
\usepackage{cleveref}

\setlist[itemize]{leftmargin=1.5em,itemsep=2pt,topsep=3pt}
\setlist[enumerate]{leftmargin=1.6em,itemsep=2pt,topsep=3pt}
\graphicspath{{./figures/}}

\algrenewcommand\algorithmicrequire{\textbf{Input:}}
\algrenewcommand\algorithmicensure{\textbf{Output:}}

\newtheorem{theorem}{Theorem}
\newtheorem{lemma}[theorem]{Lemma}

\newtheorem{corollary}[theorem]{Corollary}
\theoremstyle{definition}

\newtheorem{remark}[theorem]{Remark}

\newcommand{\eps}{\varepsilon}
\newcommand{\Hd}{\mathbb{H}^{D}}
\newcommand{\BD}{\mathbb{B}^{D}}
\newcommand{\R}{\mathbb{R}}
\newcommand{\OPT}{\mathrm{OPT}}
\newcommand{\cost}{\Phi}

\newcommand{\Ball}{B_H}

\newcolumntype{Y}{>{\raggedright\arraybackslash}X}

\DeclareMathOperator{\arcosh}{arcosh}
\DeclareMathOperator{\artanh}{artanh}

\title{Coresets for Continuous $k$-Center in Hyperbolic Space
\thanks{
This work was supported by National Research Foundation of Korea (NRF) grants funded by the Korea government (MSIT) (No. RS-2026-25471649 and No. RS-2024-00414849).
}}

\author{Eunku Park\\
Department of Liberal arts and Sciences\\
DGIST, Republic of Korea\\
\texttt{parkeun9@dgist.ac.kr}
}

\begin{document}
\maketitle

\begin{abstract}
We study coresets for the continuous $k$-center problem in fixed-dimensional hyperbolic space $\mathbb H^D$, where centers may be placed anywhere in the ambient space. For a finite input set $P$ of $n$ points and $0<\eps<1$, we construct an unweighted subset of size $(k/\eps)^{O(kD)}$ such that every exact optimal solution for the subset is a $(1+\eps)$-approximation for $P$. The construction takes $O\bigl(nkD(k/\eps)^{O(kD)}\bigr)$ time. The size is independent of $n$ and the enclosing radius of the input. After a constant-factor adjustment of the accuracy parameter, the same construction gives a strong coreset, preserving the cost of every set of at most $k$ centers.

The construction combines local Euclidean grids at bounded scales with radial shells at larger scales. On each representative sphere, hyperbolic balls induce spherical caps, and additive increases in hyperbolic distance yield multiplicative expansions of their chordal radii. This allows us to apply a recursive multiplicative-ball certificate to select a radius-independent set of input witnesses.

We prove lower bounds under the weaker requirement of transferring every exact optimal solution of the reduced instance. For $k\ge8D$ and sufficiently small $\eps$, some inputs require $\Omega_D(k\eps^{-D})$ points. For every fixed $D\ge2$, instances with $\eps=\Theta(1/k)$ require $2^{\Omega(k)}$ points, ruling out coreset-size bounds polynomial jointly in $k$ and $1/\eps$.

For inputs consisting of $k$ sufficiently separated clusters in dimension $D\ge2$, we obtain simultaneous additive and multiplicative weak coreset guarantees. When both error parameters equal $\eps$, the optimal worst-case size is $\Theta_D(k\eps^{-(D-1)/2})$ for $k\ge2$ and sufficiently small $\eps$. Under the same separation condition, no radius-independent additive weak coreset bound is possible in the Euclidean plane, even for $k=2$.
\end{abstract}

\section{Introduction}

The $k$-center problem is a basic clustering problem. Given a finite set $P$ in a metric space, the goal is to choose at most $k$ centers so as to minimize the maximum distance from an input point to its nearest center. In the \emph{continuous} version, the centers may be placed anywhere in the ambient space. They are not required to be input points. For a center set $C$, we write
\[
    \cost_P(C)=\max_{p\in P}\min_{c\in C} d(p,c),
\]
and the optimum value is the minimum of this quantity over all $|C|\le k$.

A common way to make geometric clustering faster is to build a \emph{coreset}, which is a much smaller subset, or sometimes a weighted subset, that preserves the objective value sufficiently well. This paper constructs an unweighted  coreset for continuous $k$-center in the $D$-dimensional hyperbolic space $\Hd$, where $D$ is fixed. We want every exact continuous $k$-center solution on the reduced set to be a $(1+\eps)$-approximation on the original set.

In fixed-dimensional Euclidean space, such reductions are usually based on grids. After one has a constant-factor estimate of the optimal radius, a grid of mesh proportional to $\eps$ times that radius is fine enough to transfer solutions from the reduced set back to the input. The number of occupied grid cells near the relevant part of the input is polynomial in $1/\eps$ and independent of $n$. This is the geometric reason behind many Euclidean coreset constructions~\cite{agarwal2002exact,agarwal2005geometric}.

In dimensions $D\ge 2$, hyperbolic space behaves differently. Its local geometry is Euclidean-like at small scale, but the volume of large hyperbolic balls grows exponentially with their radius. Therefore, a uniform grid of a hyperbolic ball of radius $R$ has size exponential in $R$. This is precisely the obstacle, since a direct Euclidean-style construction would give a coreset whose size depends on the spatial extent of the input. Our goal is to remove this dependence on the input radius.

Hyperbolic space is widely used in machine learning because its geometry naturally represents hierarchical and tree-like data~\cite{nickel2017poincare,nickel2018learning,chami2020trees,lin2023hyperbolic}. Many algorithmic problems in hyperbolic and Gromov-hyperbolic spaces have also been studied, including diameter and radius computation, nearest-neighbor search, and tree-like approximations~\cite{chepoi2008diameters,krauthgamer2006algorithms,kisfaludi2025quadtree,park2025embeddings}. Coresets in hyperbolic space have also been considered for farthest-point and related problems, but those constructions do not directly apply to the $k$-center objective, where the maximum is taken over distances to the nearest selected center~\cite{park2025coresets}. These observations motivate a coreset construction for continuous $k$-center in hyperbolic space. Our construction applies in every fixed dimension and preserves the cost of all continuous center configurations. We also establish lower bounds and a separation between hyperbolic and Euclidean additive guarantees.

\paragraph{Main results.}
Let $P \subset \mathbb{B}^D$ be a set of $n$ points in the Poincar\'e ball model of $\mathbb{H}^D$. We use $d_H$ for hyperbolic distance and write
\[
\operatorname{OPT}_k(P) := \min_{\substack{C\subseteq\mathbb{H}^D\\|C|\le k}} \Phi_P(C).
\]

\begin{theorem}[Main coreset theorem]
\label{thm:main}
Let $D \ge 1$ be fixed, let $k\ge1$ be an integer, let $0<\varepsilon<1$, and let $P \subset \mathbb{B}^D$ be finite and nonempty. One can construct an $\varepsilon$-coreset $P_\varepsilon \subseteq P$ such that every optimal continuous $k$-center solution
\[
C_{P_\varepsilon} \in \arg\min_{\substack{C\subseteq\mathbb{H}^D\\|C|\le k}} \Phi_{P_\varepsilon}(C)
\]
satisfies 
\[
\Phi_P(C_{P_\varepsilon}) \le (1+\varepsilon)\operatorname{OPT}_k(P).
\]
Moreover, $|P_\varepsilon|=(k/\varepsilon)^{O(kD)}$, and the construction time is $O\bigl(nkD(k/\varepsilon)^{O(kD)}\bigr)$. In particular, for fixed $D$, $k$, and $\varepsilon$, the coreset has constant size and can be constructed in $O(n)$ time.
\end{theorem}

The construction selects a subset of the original input points while allowing the centers of the reduced continuous $k$-center problem to lie anywhere in $\mathbb{H}^D$. The analysis also yields a stronger guarantee that preserves the cost of every admissible center set.

\begin{corollary}[Strong coreset guarantee]\label{cor:strong-main}
Under the assumptions of Theorem~\ref{thm:main}, one can construct a nonempty subset $Q \subseteq P$ such that
\[
\Phi_Q(C) \le \Phi_P(C) \le (1+\varepsilon)\Phi_Q(C)
\]
for every nonempty center set $C \subseteq \mathbb{H}^D$ with $|C| \le k$. Moreover, $|Q|=(k/\varepsilon)^{O(kD)}$, and the construction time is $O\bigl(nkD(k/\varepsilon)^{O(kD)}\bigr)$.
\end{corollary}

The strong guarantee follows from the same construction after a constant-factor adjustment of the accuracy parameter. Thus, the asymptotic size and construction-time bounds remain unchanged.

We complement these upper bounds with lower bounds that apply even to unweighted input subsets satisfying only the weak solution-transfer guarantee of Theorem~\ref{thm:main}.

\begin{theorem}[Lower bounds for weak coresets]\label{thm:lower-main}
The following statements hold for continuous $k$-center in hyperbolic space.
\begin{enumerate}
\item For every fixed $D \ge 1$, every $k \ge 2$, and all sufficiently small $\varepsilon>0$, there exists a finite input $P \subset \mathbb{H}^D$ such that every weak $\varepsilon$-coreset $Q \subseteq P$ has size $|Q|=\Omega_D\bigl(k\varepsilon^{-(D-1)/2}\bigr)$.

\item For every $D \ge 1$, $k \ge 8D$, and $0<\varepsilon\le 1/128D$, there exists a finite input $P \subset \mathbb{H}^D$ such that every weak $\varepsilon$-coreset $Q \subseteq P$ has size $|Q|=\Omega_D(k\varepsilon^{-D})$.

\item For every fixed $D \ge 2$, there exist instances with $\varepsilon=\Theta(1/k)$ for which every weak $\varepsilon$-coreset has size $2^{\Omega(k)}$. Consequently, no coreset-size bound polynomial jointly in $k$ and $1/\varepsilon$ is possible for all finite inputs in $\mathbb{H}^D$.
\end{enumerate}
All three lower bounds hold even when the centers of the reduced instance may be placed anywhere in $\mathbb{H}^D$.
\end{theorem}

The first lower bound arises from angular witness configurations, while the second uses a local-dimensional construction. The third combines a tree-based obstruction with a low-distortion embedding into hyperbolic space. In each construction, every point of the input is necessary for the weak solution-transfer guarantee. The joint-parameter obstruction does not rule out polynomial dependence on $k$ when $\varepsilon$ is fixed.

We next identify a class of inputs for which additive and multiplicative guarantees can be obtained simultaneously. For a finite nonempty set $A$, let $\operatorname{rad}(A)$ denote its minimum enclosing radius, and for two finite nonempty sets $A,B$, let $\operatorname{dist}(A,B)$ denote their minimum distance. Suppose that $P$ admits a partition into exactly $k$ nonempty subsets $P=P_1\mathbin{\dot\cup}\cdots \mathbin{\dot\cup}P_k$ such that, writing $R:=\max_{1\le i\le k}\operatorname{rad}(P_i)$, we have $\operatorname{dist}(P_i,P_j)>2R$ for all $i \ne j$.

\begin{theorem}\label{thm:separated-main}
Let $D \ge 2$ be fixed, and suppose that a finite set $P \subset \mathbb{H}^D$ admits the separated partition described above. For $\alpha>0$ and $0<\varepsilon<1$, put $\eta:=\min\left\{\alpha, \frac{\varepsilon}{1+\varepsilon}\right\}$. There exists a nonempty subset $Q \subseteq P$ of size $|Q|=O_D\bigl(k\eta^{-(D-1)/2}\bigr)$ such that every exact optimal continuous $k$-center solution $C_Q$ for $Q$ satisfies
\[
\Phi_P(C_Q) \le \min\left\{\operatorname{OPT}_k(P)+\alpha, (1+\varepsilon)\operatorname{OPT}_k(P) \right\}.
\]
The size bound is independent of both $|P|$ and $R$.

Moreover, when $\alpha=\varepsilon$, the optimal worst-case coreset size is
\[
\Theta_D\bigl(k\varepsilon^{-(D-1)/2} \bigr)
\]
for every $k \ge 2$ and all sufficiently small $\varepsilon>0$.
\end{theorem}

The separation condition ensures that every exact optimal solution of the reduced instance assigns exactly one center to each nonempty cluster. This allows the hyperbolic farthest-point coreset construction to be applied independently within the clusters. The resulting bound is independent of the cluster radius, even for the additive guarantee.

The corresponding radius-independent additive guarantee is impossible in Euclidean space, even under the same separation condition.

\begin{theorem}\label{thm:euclidean-main}
For every $\alpha>0$ and every integer $N \ge 3$, there exists a set $P \subset \mathbb{R}^2$ of $N+2$ points admitting a partition into two nonempty subsets $P_1,P_2$ satisfying
\[
\operatorname{dist}(P_1,P_2)>2R, \qquad R:=\max\{\operatorname{rad}(P_1), \operatorname{rad}(P_2)\},
\]
such that every nonempty subset $Q \subseteq P$ for which every exact optimal continuous $2$-center solution $C_Q$ satisfies
\[
\Phi_P(C_Q) \le \operatorname{OPT}_2(P)+\alpha
\]
must equal $P$.

Consequently, no radius-independent additive weak coreset-size bound depending only on $k$, $D$, and $\alpha$ is possible, even for $k=D=2$.
\end{theorem}

Together, Theorems~\ref{thm:separated-main} and~\ref{thm:euclidean-main} show that the separation condition guarantees radius-independent additive weak coresets in hyperbolic space, but not in Euclidean space.

\subparagraph{Our approach.}
Our construction combines local grids at bounded scales with radial shells at larger scales. For $k<|P|$, we first run Gonzalez's farthest-first traversal~\cite{gonzalez1985clustering}, obtaining anchors $A_k=\{a_1,\ldots,a_k\} \subseteq P$ and the reference scale $R_G=\cost_P(A_k)$. The case $k \ge |P|$ is immediate by returning $P$. Assigning each input point to a nearest anchor gives clusters $P_i \subseteq B_H(a_i,R_G)$. We use the constant-factor bound $\OPT_k(P)\le R_G\le 4\OPT_k(P)$. Fix an absolute threshold $R_0>0$, chosen in advance independently of the input and the accuracy parameter.

When $R_G \le R_0$, we map each anchor to the origin by a hyperbolic isometry and place a Euclidean grid over its cluster in the Poincar\'e ball model. Hyperbolic and Euclidean distances are uniformly comparable within this fixed-radius region. A grid of mesh $\Theta(\eps R_G)$, with a sufficiently small constant, allows us to retain one original input point per occupied cell so that every point is within hyperbolic distance $\eps R_G/4$ of its representative. The total number of retained points is $O_D(k\eps^{-D})$.

When $R_G>R_0$, we avoid an ambient grid by partitioning each cluster into $O(1/\eps)$ radial shells of width at most $\Delta:=\eps R_G/8$. For a shell with radial interval $[r_-,r_+]$, let $\bar r:=(r_-+r_+)/2$ be its midpoint radius. We project the points of this shell radially onto the hyperbolic sphere of radius $\bar r$ centered at the corresponding anchor, moving every point by at most $\Delta/2$. In polar coordinates around this anchor, a point on the representative sphere is written as $(\bar r,\xi)$, where $\xi \in \mathbb S^{D-1}$. The hyperbolic law of cosines shows that the directions of points on this sphere lying in a hyperbolic ball form a spherical cap. Consequently, for any center set $C \subseteq \Hd$ with $1 \le |C| \le k$ and any $t \ge 0$, the set $\left\{ \xi \in \mathbb S^{D-1} \mid d_{H}\bigl((\bar r,\xi),C\bigr)\le t \right\}$ is a union of at most $k$ spherical caps.

We select a small subset of the projected directions using a recursive multiplicative-ball certificate, following the Euclidean construction in~\cite[Theorem~6.1]{agarwal2005geometric}. If at most $k$ caps cover the selected directions, multiplying their chordal radii by $1+\gamma$ makes them cover all projected directions. The key hyperbolic property is that this multiplicative expansion can be absorbed by an additive increase in the distance threshold. Specifically, choosing $\gamma_G:=\min\{1,e^{\Delta/2}-1\}$ requires an increase of at most $\Delta$. Since $R_G>R_0$, we have $\gamma_G=\Omega(\eps)$, so the certificate contains at most $(k/\eps)^{O(kD)}$ directions per shell, independently of $R_G$. Retaining the corresponding original input points and accounting for both radial projections gives additive error at most $2\Delta=\eps R_G/4$ in the maximum distance from the shell to every center set.

Retaining all anchors together with these representatives gives $P_\eps \subseteq P$. In both regimes, the resulting set satisfies $\cost_P(C) \le \cost_{P_\eps}(C)+\frac{\eps R_G}{4}$ for every $C \subseteq \Hd$ with $1 \le |C| \le k$. Thus, every exact optimal continuous $k$-center solution $C_{P_\eps}$ for the reduced instance satisfies $\cost_P(C_{P_\eps}) \le \OPT_k(P_\eps)+\frac{\eps R_G}{4} \le (1+\eps)\OPT_k(P)$. The total coreset size is $(k/\eps)^{O(kD)}$, and the construction time is $O\bigl(nkD(k/\eps)^{O(kD)}\bigr)$. Running the same construction with accuracy $\eps/(1+\eps)$ also yields the strong coreset guarantee from the uniform additive inequality above.

\subparagraph{Relation to previous work.}
The $k$-center problem has been studied extensively in both general metric spaces and geometric spaces. In a general metric space, Gonzalez's farthest-first traversal gives a $2$-approximation for the discrete $k$-center problem~\cite{gonzalez1985clustering}. This approximation factor is best possible in polynomial time unless $\mathrm{P}=\mathrm{NP}$~\cite{hochbaum1985best}. The discrete and continuous versions should be distinguished.  In the discrete problem, the centers must be chosen from the input point set, whereas in the continuous problem the centers may be placed anywhere in the ambient space. Our construction uses Gonzalez's traversal only to obtain
a constant-factor reference scale. By exploiting fixed-dimensional hyperbolic geometry, we construct an $\eps$-coreset for continuous $k$-center. The running-time bound in Theorem~\ref{thm:main} applies only to coreset construction and excludes the time required to solve the reduced instance.

In fixed-dimensional Euclidean space, the continuous $k$-center problem admits efficient approximation schemes. Agarwal and Procopiuc gave algorithms for Euclidean $k$-center based on geometric discretization and finite sets of candidate centers~\cite{agarwal2002exact}. More recently, fine-grained lower bounds have clarified the dependence on $k$, $\eps$, and the dimension for Euclidean continuous $k$-center~\cite{blank2026fine}. In fixed-dimensional Euclidean space, a standard grid-based coreset construction uses a mesh proportional to $\eps$ times a constant-factor estimate of the optimal radius and retains one representative point from each occupied cell. Coresets of size independent of the number of input points are known for many Euclidean clustering and extent problems~\cite{agarwal2005geometric,feldman2011unified}. However, the grid-based approach does not directly extend to hyperbolic space in dimensions $D \ge 2$. The volume of a large hyperbolic ball grows exponentially with its radius. A uniform discretization at a fixed hyperbolic resolution therefore requires exponentially many cells as the radius increases. We overcome this obstacle by combining local grids at bounded scales with radial shells and spherical-cap coverage certificates at larger scales. The resulting coreset size is independent of both the number of input points and the enclosing radius.

For planar graph metrics, Bourneuf and Pilipczuk~\cite[Theorem~6.1]{BourneufPilipczuk2025} give strong $k$-center coresets of size $(k+1)^{\eps^{-O(1)}}$. Kluk et al.~\cite[Theorem~1.10]{KlukEtAl2026} obtain the complementary bound $(1/\eps)^{2^{O(k)}}$. After adjusting the accuracy parameters, both bounds transfer to continuous $k$-center in $\mathbb H^2$ through the argument in Appendix~\ref{app:planar-transfer}. This transfer establishes the existence of strong coresets but does not directly provide an efficient construction. Our algorithm constructs strong coresets of size $(k/\eps)^{O(kD)}$ in every fixed dimension. In dimension two, the exponent of $1/\eps$ is linear in $k$ rather than exponential as in the bound of Kluk et al. The existence bound obtained from Bourneuf and Pilipczuk has better dependence on $k$ when $\eps$ is fixed.

For lower bounds, Kluk et al.~\cite[Theorems~1.11 and~1.12]{KlukEtAl2026} establish a joint-parameter obstruction for strong coresets in tree metrics when $k=\Theta(1/\eps)$. We combine their tree construction with Sarkar's low-distortion embedding~\cite{Sarkar2011} to obtain a lower bound in hyperbolic space. Additional anchor points allow a center set exposing an omitted point to be completed to an exact optimal solution of the reduced instance. This establishes the obstruction even under the weak solution-transfer guarantee.

Park and Vigneron~\cite{park2025coresets} constructed farthest-point coresets of size $O_D(\eps^{-D})$ in fixed-dimensional hyperbolic space. Their coresets provide simultaneous additive and relative-error guarantees for every query point. A subsequent ideal-boundary approach~\cite{Park2026Boundary} establishes the optimal worst-case size $\Theta_D(\eps^{-(D-1)/2})$ for $D\ge 2$. Both farthest-point coreset size bounds are independent of the number of input points and the enclosing radius.

These farthest-point results approximate the farthest-distance function $\max_{p\in P}d_H(p,c)$ for every individual center $c$. Such a guarantee does not directly extend to the $k$-center objective $\max_{p\in P}\min_{c\in C}d_H(p,c)$ because the nearest-center assignment depends on the entire center set $C$. We address this difficulty using local grids at bounded scales and radial-shell projections with coverage certificates for unions of at most $k$ spherical caps at larger scales. For the separated inputs of Section~\ref{sec:sep-inputs}, we apply the farthest-point coreset theorem~\cite{Park2026Boundary} independently within each cluster.

To the best of our knowledge, this is the first construction of $\eps$-coresets for continuous $k$-center with general $k \ge 2$ that works in every fixed-dimensional hyperbolic space and has size independent of both the number of input points and the enclosing radius.


\section{Preliminaries}\label{sec:preliminaries}
\subsection{Metric notation and coresets}
Let $(X,d)$ be a metric space, and let $P \subset X$ be a finite nonempty point set. For a finite center set $C\subseteq X$, define
\[
    d(p,C) := \min_{c\in C}d(p,c), \qquad \cost_P(C) := \max_{p\in P}d(p,C).
\]
We use the convention $d(p,\varnothing)=+\infty$.
The continuous \(k\)-center optimum is
\[
    \OPT_k(P) := \min_{\substack{C\subseteq X\\ |C|\le k}} \cost_P(C).
\]
We apply this definition in settings where the minimum is attained. In particular, this holds for finite inputs in $\Hd$,
since $\Hd$ is a proper metric space.

A nonempty subset $P_\eps \subseteq P$ is an \emph{$\eps$-coreset} for continuous $k$-center if
\[
    \cost_P(C_\eps) \le (1+\eps)\OPT_k(P) \qquad \text{for every} \qquad C_\eps \in \arg\min_{\substack{C\subseteq X\\ |C|\le k}} \cost_{P_\eps}(C).
\]

Unless explicitly stated otherwise, an $\eps$-coreset refers to this weak solution-transfer guarantee. Strong coresets are considered in Section~\ref{subsec:strong-coreset}.

Thus, the coreset is a subset $P_\eps \subseteq P$ such that every exact optimum of the reduced instance is a $(1+\eps)$-approximate solution for the original point set $P$. This definition requires every exact optimal solution of the reduced instance to transfer to the original input. Preserving only the optimum value does not express this solution-transfer requirement.

\subsection{The Poincar\'e ball model}
We work in the Poincar\'e ball model of $D$-dimensional hyperbolic space with curvature $-1$. Its underlying point set is
\[
    \BD=\{x\in\R^D:\|x\|<1\},
\]
where $\|\cdot\|$ denotes the Euclidean norm. We denote the origin by $O=0$. The hyperbolic distance is
\[
    d_{H}(x,y) = \arcosh\left( 1+\frac{2\|x-y\|^2}{(1-\|x\|^2)(1-\|y\|^2)} \right),
\]
and, in particular,
\[
    d_{H}(O,x)=2\artanh\|x\|.
\]

For every $a\in\BD$, we choose a hyperbolic isometry $\mu_a:\BD\to\BD$ satisfying $\mu_a(a)=O$; see~\cite[Theorem~4.5.2]{ratcliffe2006foundations}. Since $\mu_a$ preserves hyperbolic distances,
\[
    d_{H}(x,y) = d_{H}\bigl(\mu_a(x),\mu_a(y)\bigr) \qquad\text{for all }x,y\in\BD.
\]
Consequently, for any finite nonempty point set $S \subseteq \BD$ and finite nonempty center set $C \subseteq \BD$,
\[
    \cost_{\mu_a(S)}\bigl(\mu_a(C)\bigr) = \cost_S(C).
\]
Thus, when analyzing a cluster with anchor $a$, we may move $a$ to the origin by applying $\mu_a$ to the cluster and interpreting every center set in the same transformed coordinates.

\begin{figure}
    \centering
    \includegraphics[width=5in]{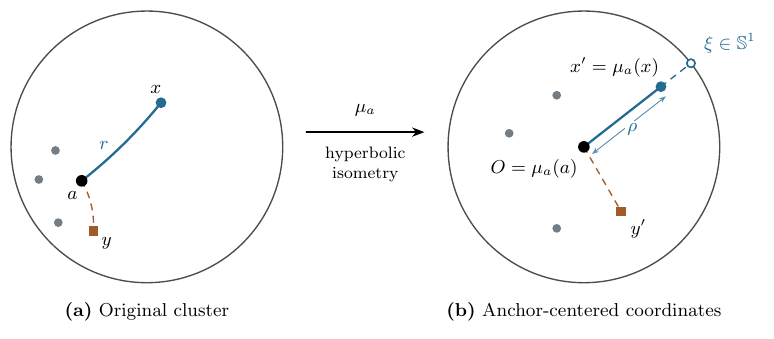}
    \caption{Anchor-centered coordinates in the Poincar\'e disk ($D=2$). (a) A cluster with anchor $a$. (b) Its image under the hyperbolic isometry $\mu_a$, which sends $a$ to the origin $O$ and preserves all hyperbolic distances. Writing $x'=\mu_a(x)=\rho\xi$ and $y'=\mu_a(y)$, the geodesic from $a$ to $x$ becomes a radial segment of Euclidean length $\rho$ and hyperbolic length $r=2\artanh\rho$. The open boundary marker $\xi\in\mathbb S^1$ represents its direction.}
    \label{fig:poincare-recentering}
\end{figure}

For $x \in \BD \setminus \{a\}$, write
\[
    x':=\mu_a(x)=\rho\xi, \qquad \rho:=\|x'\|, \qquad \xi:=\frac{x'}{\|x'\|}\in\mathbb S^{D-1}.
\]
The hyperbolic radius relative to $a$ is
\[
    r:=d_{H}(a,x)=d_{H}(O,x')=2\artanh\rho, \qquad \rho=\tanh(r/2).
\]
We therefore represent $x$ in anchor-centered coordinates by $(r,\xi)$, so that $\mu_a(x)=\tanh(r/2)\xi$; see Figure~\ref{fig:poincare-recentering}. The anchor itself has radius zero and no specified direction.

All running times are stated in an augmented real-RAM model with constant-time arithmetic and comparisons, as well as constant-time evaluation of square roots, exponentials, logarithms, and the hyperbolic and inverse hyperbolic functions
used below. For fixed $D$, evaluating $\mu_a$ or $\mu_a^{-1}$ at a point, computing a hyperbolic distance, and converting a non-anchor point to anchor-centered polar coordinates each take $O(1)$ time in this model. We also assume constant-time floor operations and direct addressing for the bounded grids used below.

\subsection{Hyperbolic spheres and spherical caps}\label{subsec:hyperbolic-caps}
Throughout this subsection, polar coordinates are taken around a common origin $O$, and $\langle\cdot,\cdot\rangle$ denotes the Euclidean inner product of direction vectors. We write
\[
    B_H(c,t):=\{x\in\Hd:d_H(x,c)\le t\}
\]
for the closed hyperbolic ball of radius $t\ge 0$ centered at $c$. We first express the hyperbolic law of cosines in polar coordinates. It follows that the intersection of a hyperbolic ball with a sphere of fixed radius centered at $O$ corresponds to a spherical cap under the angular map $(r,\xi)\mapsto\xi$. We include the derivation for completeness.

\begin{lemma}\label{lem:polar-cosine-law}
Let $x,y \in \Hd$, and put $r=d_H(O,x)$ and $s=d_H(O,y)$. If $r,s>0$, write $x=(r,\xi)$ and $y=(s,\zeta)$ in polar coordinates around $O$. Then 
\[
    \cosh d_H(x,y) = \cosh r\cosh s - \sinh r\sinh s\,\langle\xi,\zeta\rangle.
\]
Equivalently,
\[
    \cosh d_H(x,y) = \cosh(r-s) + \sinh r\sinh s \bigl(1-\langle\xi,\zeta\rangle\bigr).
\]
If $r=0$ or $s=0$, then $d_H(x,y)=r+s$.
\end{lemma}

\begin{proof}
If $r=0$ or $s=0$, one of the points is $O$, and the stated distance follows immediately. Assume that $r,s>0$, and let
\[
    \theta:=\arccos\langle\xi,\zeta\rangle
\]
be the angle at $O$ between the geodesic segments to $x$ and $y$. For $0<\theta<\pi$, the hyperbolic law of cosines~\cite{anderson2005hyperbolic,ratcliffe2006foundations} gives
\[
    \cosh d_H(x,y) = \cosh r\cosh s-\sinh r\sinh s\cos\theta.
\]
When $\theta=0$ or $\theta=\pi$, the points lie on the same or opposite rays from $O$, respectively. Their distance is then $|r-s|$ or $r+s$, so the same identity holds in these cases as well.

Since $\cos\theta=\langle\xi,\zeta\rangle$, the first formula follows. The second follows from
\[
    \cosh(r-s)=\cosh r\cosh s-\sinh r\sinh s.
\]
\end{proof}

A \emph{spherical cap} is a subset of $\mathbb S^{D-1}$ of the form
\[
    K(\zeta,\tau) := \left\{ \xi\in\mathbb S^{D-1}: \langle\xi,\zeta\rangle\ge\tau \right\},
\]
where $\zeta\in\mathbb S^{D-1}$ and $\tau\in\mathbb R$. We allow degenerate caps. The set $K(\zeta,\tau)$ is empty when $\tau>1$, equals $\{\zeta\}$ when $\tau=1$, and is the whole sphere when $\tau\le -1$.

\begin{lemma}\label{lem:ball-section-cap}
Fix $r>0$, let $c\in\Hd$, and let $t\ge 0$. Then
\[
    K_r(c,t) := \left\{ \xi\in\mathbb S^{D-1}: d_H\bigl((r,\xi),c\bigr)\le t \right\}
\]
is a spherical cap.

More precisely, put $s:=d_H(O,c)$. If $s>0$, write $c=(s,\zeta)$ in polar coordinates around $O$. Then
\[
    K_r(c,t)=K\bigl(\zeta,\tau(r,s,t)\bigr), \qquad \tau(r,s,t) := \frac{\cosh r\cosh s-\cosh t}{\sinh r\sinh s}.
\]
If $s=0$, then
\[
    K_r(c,t) =
    \begin{cases}
        \varnothing, & t<r,\\
        \mathbb S^{D-1}, & t\ge r.
    \end{cases}
\]
\end{lemma}

\begin{proof}
If $s=0$, then $c=O$ and $d_H\bigl((r,\xi),O\bigr)=r$ for every $\xi\in\mathbb S^{D-1}$. This gives the stated empty-set and whole-sphere cases.

Now assume that $s>0$. By Lemma~\ref{lem:polar-cosine-law},
\[
    \cosh d_H\bigl((r,\xi),c\bigr) = \cosh r\cosh s - \sinh r\sinh s\,\langle\xi,\zeta\rangle.
\]
Since $\cosh$ is strictly increasing on $[0,\infty)$ and $\sinh r\sinh s>0$, we have
\[
\begin{aligned}
    d_H\bigl((r,\xi),c\bigr)\le t 
    &\Longleftrightarrow \cosh r\cosh s - \sinh r\sinh s\,\langle\xi,\zeta\rangle \le \cosh t\\
    &\Longleftrightarrow \langle\xi,\zeta\rangle \ge \frac{\cosh r\cosh s-\cosh t}{\sinh r\sinh s}.
\end{aligned}
\]
Thus $K_r(c,t)=K\bigl(\zeta,\tau(r,s,t)\bigr)$, which is a spherical cap under the convention above.
\end{proof}

Figure~\ref{fig:hyperbolic-sphere-cap} illustrates the nondegenerate case of Lemma~\ref{lem:ball-section-cap} for $D=2$.

\begin{figure}
    \centering
    \includegraphics[width=5in]{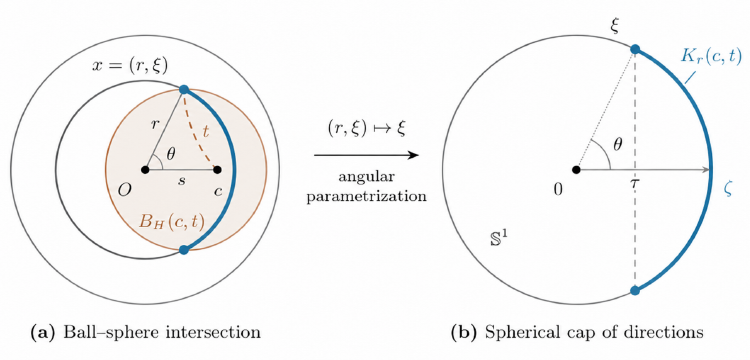}
    \caption{From a hyperbolic ball--sphere intersection to a spherical cap for $D=2$. (a) In the Poincar\'e disk, the shaded closed ball $B_H(c,t)$, with $c=(s,\zeta)$, intersects the hyperbolic sphere of radius $r$ centered at $O$ in the highlighted arc. (b) Under the angular map $(r,\xi)\mapsto\xi$, this arc corresponds to the cap $K_r(c,t)$ on $\mathbb S^1$. The dashed chord joins the two boundary directions satisfying $\langle\xi,\zeta\rangle=\tau(r,s,t)$. At the marked boundary point $x=(r,\xi)$, we have $d_H(x,c)=t$ and $\cos\theta=\tau(r,s,t)$, where $\theta = \arccos \langle \xi,\zeta\rangle$.}
    \label{fig:hyperbolic-sphere-cap}
\end{figure}

\begin{corollary}\label{cor:k-cap-union}
Fix $r>0$, let $C\subseteq\Hd$ satisfy $|C|\le k$, and let $t\ge 0$. Then
\[
    \left\{\xi\in\mathbb S^{D-1}: d_H\bigl((r,\xi),C\bigr)\le t \right\}
\]
is a union of at most $k$ spherical caps.
\end{corollary}

\begin{proof}
If $C=\varnothing$, the sublevel set is empty by the convention $d_H(x,\varnothing)=+\infty$. It is therefore the union of zero caps.

Otherwise, for every $\xi\in\mathbb S^{D-1}$,
\[
    d_H\bigl((r,\xi),C\bigr)\le t \quad\Longleftrightarrow\quad d_H\bigl((r,\xi),c\bigr)\le t \text{ for some }c\in C.
\]
Consequently,
\[
    \left\{\xi\in\mathbb S^{D-1}: d_H\bigl((r,\xi),C\bigr)\le t \right\} = \bigcup_{c\in C}K_r(c,t).
\]
Each $K_r(c,t)$ is a spherical cap by Lemma~\ref{lem:ball-section-cap}, and there are at most $k$ such sets.
\end{proof}

\subsection{Farthest-first anchors}
Let $n=|P|$, and let $k\ge 1$ be an integer. If $k\ge n$, we set $A_k=P$ and $R_G=0$ and use the singleton partition of $P$.

For $1\le k<n$, Gonzalez's farthest-first traversal~\cite{gonzalez1985clustering} starts with an arbitrary point $a_1\in P$ and sets $A_1=\{a_1\}$. For $i=2,\ldots,k$, it chooses
\[
    a_i\in \arg\max_{p\in P} d_H(p,A_{i-1}), \qquad A_i:=A_{i-1}\cup\{a_i\}.
\]
The selected points are called \emph{anchors}. We define
\[
    R_G:=\cost_P(A_k)
\]
and assign each input point to a nearest anchor. Ties in both selection and assignment are resolved by fixed rules. Let $P_i$ be the set of points assigned to $a_i$.
Then
\[
    P=P_1\mathbin{\dot\cup}\cdots\mathbin{\dot\cup}P_k, \qquad P_i\subseteq\Ball(a_i,R_G).
\]

\begin{lemma}\label{lem:gonzalez}
The farthest-first radius satisfies
\[
    \OPT_k(P)\le R_G\le 4\OPT_k(P).
\]
For fixed $D$, the anchor set $A_k$, the radius $R_G$, and the nearest-anchor partition can be computed in $O(kn)$ time under the computational model stated above.
\end{lemma}

\begin{proof}[Proof sketch]
The lower bound follows because $A_k$ is a feasible set of continuous centers. If $k<n$, the $k$ anchors together with a point farthest from $A_k$ form $k+1$ points whose pairwise distances are at least $R_G$. Two of these points must be assigned to the same center of an optimal continuous solution. The triangle inequality therefore gives $R_G\le 2\OPT_k(P)$, which implies the stated upper bound. The case $k\ge n$ has $R_G=\OPT_k(P)=0$.

Maintaining the distance and a nearest-anchor index for each input point takes $O(kn)$ time. The full proof is given in Appendix~\ref{app:proof-gonzalez}.
\end{proof}


\section{Bounded Scale: Local Grids}\label{sec:bounded}
Let $D\ge 1$ be fixed, and let $0<\eps<1$. Assume that $R_G\le R_0$, where $R_0>0$ is a fixed absolute constant independent of $n$, $k$, and $\eps$. If $R_G=0$, every input point coincides with an anchor. Thus $|P|\le k$, and we return the exact coreset $P_\eps=P$.

Now assume that $0<R_G\le R_0$. For each anchor cluster $P_i$, the isometry $\mu_{a_i}$ gives
\[
    \mu_{a_i}(P_i) \subseteq \Ball(O,R_G) \subseteq \Ball(O,R_0).
\]
In the Poincar\'e ball model, $\Ball(O,R_0)$ is the Euclidean ball centered at $O$ with radius
\[
    \rho_0:=\tanh(R_0/2).
\]
Set
\[
    L_0:=\frac{2}{1-\rho_0^2}.
\]
On this ball, hyperbolic and Euclidean distances satisfy
\[
    2\|x-y\|\le d_H(x,y)\le L_0\|x-y\|.
\]
The constant $L_0$ depends only on $R_0$.

Choose a Euclidean grid of half-open cubes with side length
\[
    h:=\frac{\eps R_G}{4L_0\sqrt D} = \Theta(\eps R_G).
\]
Two points of $\Ball(O,R_0)$ in the same cell have Euclidean distance at most $h\sqrt D$ and hence hyperbolic distance at most $\eps R_G/4$.

For each cluster $P_i$, retain one original input point from each occupied cell of $\mu_{a_i}(P_i)$. For the cell containing $O$, retain $a_i$. Let $P_\eps$ be the union of the retained points. Since $\mu_{a_i}(P_i)\subseteq \Ball(O,R_G)$ and
\[
    \rho_G:=\tanh(R_G/2)\le R_G/2,
\]
the number of occupied cells per cluster is
\[
    O\left(\left(1+\frac{\rho_G}{h}\right)^D\right) = O(\eps^{-D}).
\]
The following theorem summarizes the resulting guarantees. Its proof is given in Appendix~\ref{app:bounded-scale}.

\begin{theorem}[Bounded-scale coreset]\label{thm:bounded-scale}
Assume that $R_G\le R_0$, where $R_0>0$ is a fixed absolute constant. For fixed $D$ and $0<\eps<1$, the construction returns $A_k\subseteq P_\eps\subseteq P$ such that
\[
    \cost_{P_\eps}(C) \le \cost_P(C) \le \cost_{P_\eps}(C)+\frac{\eps R_G}{4}
\]
for every center set $C\subseteq\Hd$ with $1\le |C|\le k$. Consequently, every exact optimal continuous $k$-center solution $C_{P_\eps}$ for $P_\eps$ satisfies
\[
    \cost_P(C_{P_\eps}) \le (1+\eps)\OPT_k(P).
\]
Moreover,
\[
    |P_\eps|=O(k/\eps^D),
\]
and the construction time is
\[
    O\left(kn+n+\frac{k}{\eps^D}\right).
\]
The hidden constants depend only on $D$ and $R_0$. In particular, for fixed $D$, $k$, and $\eps$, the coreset has constant size and can be constructed in $O(n)$ time.
\end{theorem}

The bounded-scale construction uses a uniform metric comparison on $\Ball(O,R_0)$. The next section treats larger scales without introducing a dependence on $R_G$ in the coreset-size bound.


\section{Large Scale: Radial Shells and Spherical-Cap Certificates}\label{sec:radial-shells}
The large-scale construction applies a Euclidean multiplicative ball certificate to input directions on representative spheres. Theorem~\ref{thm:shell-reduction} justifies this reduction by converting multiplicative cap enlargement into a uniform hyperbolic additive error independent of the representative radius.

Assume that $R_G\ge R_0$. For each $p\in P_i\setminus A_k$, let $(r_i(p),\xi_i(p))$ denote its polar coordinates centered at $a_i$, as defined in Section~\ref{sec:preliminaries}. Then $0<r_i(p)\le R_G$.

Put
\[
    \Delta:=\frac{\eps R_G}{8}, \qquad L:=\left\lceil\frac{8}{\eps}\right\rceil, \qquad b_j:=\min\{j\Delta,R_G\} \quad (j=0,\ldots,L).
\]
Partition $[0,R_G]$ into $I_j=[b_{j-1},b_j)$ for $j<L$ and $I_L=[b_{L-1},b_L]$. Each interval has positive length at most $\Delta$. Let $\bar r_j:=(b_{j-1}+b_j)/2$ and define
\[
    B(i,j) := \{p\in P_i\setminus A_k:r_i(p)\in I_j\}.
\]
For $p\in B(i,j)$, define its radial projection in the original space by
\[
    \pi_j(p) := \mu_{a_i}^{-1} \left( \tanh(\bar r_j/2)\xi_i(p) \right).
\]
In polar coordinates centered at $a_i$, this projection is $(\bar r_j,\xi_i(p))$. Since radial lines through the origin are hyperbolic geodesics,
\begin{equation}
    d_H\bigl(p,\pi_j(p)\bigr) = |r_i(p)-\bar r_j| \le\frac{\Delta}{2}.
    \label{eq:large-radial-projection}
\end{equation}
Every non-anchor input point belongs to exactly one shell, and there are at most $kL$ nonempty shells. The projection reduces each shell to a coverage problem on its input directions in $\mathbb S^{D-1}$.

Figure~\ref{fig:two-scale-coreset} compares the local-grid construction of Section~\ref{sec:bounded} with the radial-shell and spherical-cap construction developed here.

\begin{figure}[t]
    \centering
    \includegraphics[width=5in]{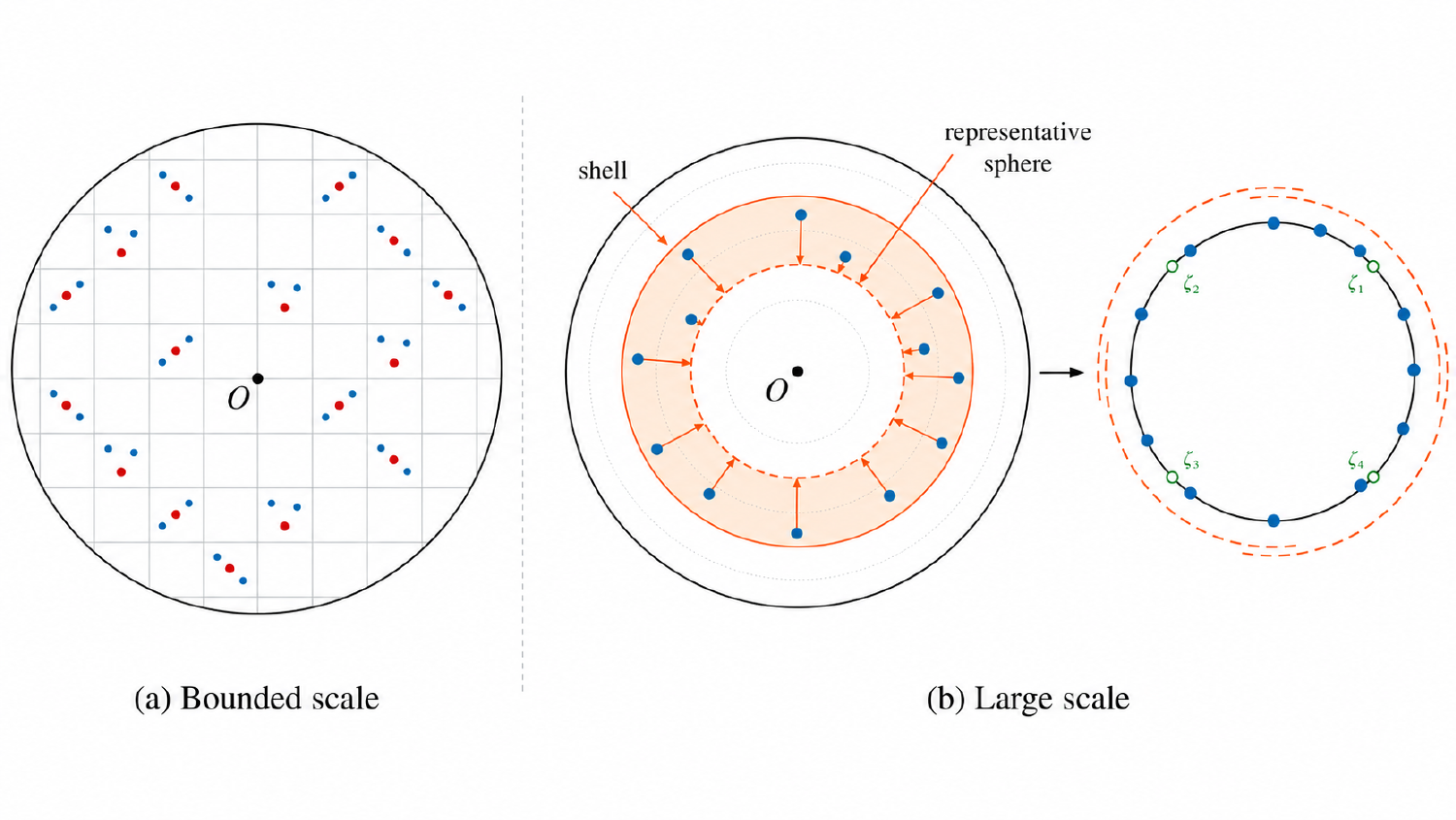}
    \caption{The two coreset constructions illustrated schematically for $D=2$ in anchor-centered coordinates. (a) A local Euclidean grid retains one original input point per occupied cell. Red points are retained representatives, and the anchor $O$ represents its own cell. (b) Points in a radial shell are projected onto a representative sphere (orange dashed circle) while preserving their directions. The rightmost circle shows the input directions in blue. Green markers $\zeta_\ell$ denote cap-axis directions, not selected input points. Orange dashed arcs illustrate enlarged caps covering the input directions. The certificate selects input directions, and the corresponding original input points are retained.}
    \label{fig:two-scale-coreset}
\end{figure}

\subsection{Ball--Sphere Intersections and Euclidean Certificates}
Fix an anchor and a representative radius $\bar r>0$. In this subsection, both points and queried centers are expressed
in the same anchor-centered coordinates. For $c\in\Hd$ and $t\ge 0$, let
\[
    K_{\bar r}(c,t) := \{\xi\in\mathbb S^{D-1}: d_H\bigl((\bar r,\xi),c\bigr)\le t\}.
\]
Put $s:=d_H(O,c)$. If $s>0$, write $c=(s,\zeta)$. Lemma~\ref{lem:polar-cosine-law} gives
\[
    \cosh d_H\bigl((\bar r,\xi),c\bigr) = \cosh(\bar r-s) + \sinh\bar r\,\sinh s \bigl(1-\langle\xi,\zeta\rangle\bigr).
\]
For $s>0$ and $t\ge|\bar r-s|$, define
\[
    \rho_c(t;\bar r)^2 := \frac{2\bigl(\cosh t-\cosh(\bar r-s)\bigr)}{\sinh\bar r\,\sinh s}, \qquad \rho_c(t;\bar r)\ge 0.
\]
Since $\|\xi-\zeta\|_2^2=2(1-\langle\xi,\zeta\rangle)$,
\begin{equation}
    K_{\bar r}(c,t) = B_2\bigl(\zeta,\rho_c(t;\bar r)\bigr) \cap\mathbb S^{D-1},
    \label{eq:large-cap-representation}
\end{equation}
where $B_2(z,\rho)$ denotes a closed Euclidean ball. For $s>0$, the cap is empty when $t<|\bar r-s|$ and is the whole sphere when $\rho_c(t;\bar r)\ge 2$. We use the untruncated formula for $\rho_c(t;\bar r)$ even when it exceeds $2$. If $s=0$, the cap is empty for $t<\bar r$ and is the whole sphere for $t\ge\bar r$.

Consequently, for every $C\subseteq\Hd$ with $1\le |C|\le k$,
\[
    \{\xi\in\mathbb S^{D-1}: d_H\bigl((\bar r,\xi),C\bigr)\le t\} = \bigcup_{c\in C}K_{\bar r}(c,t).
\]
This is an exact representation of the sublevel set by at most $k$ spherical caps, not a comparison between hyperbolic and Euclidean distances. Even at a common threshold $t$, the chordal cap radii can differ because the centers have different distances from the anchor.

We therefore use the multiplicative Euclidean coreset construction that allows covering balls of different radii and enlarges each radius separately~\cite[Theorem~6.1]{agarwal2005geometric}. The following lemma is its restriction to the direction sphere. The net-based implementation and explicit bounds are established in Appendix~\ref{app:proof-cap-certificate}.

\begin{lemma}\label{lem:algorithmic-cap-certificate}
Let $X\subseteq\mathbb S^{D-1}$ be finite and nonempty, let $k\ge 1$ be an integer, and let $0<\gamma\le 1$. One can construct a nonempty subset $W_k(X,\gamma)\subseteq X$ such that
\[
    |W_k(X,\gamma)|\le T_{k,D}(\gamma), \qquad T_{k,D}(\gamma) := k!\left(\frac{A}{\gamma}\right)^{kD},
\]
where $A>1$ is a universal constant. For every $1\le m\le k$, every $\zeta_1,\ldots,\zeta_m\in\mathbb S^{D-1}$, and every $\rho_1,\ldots,\rho_m\ge 0$, define
\[
    K_\ell:=B_2(\zeta_\ell,\rho_\ell)\cap\mathbb S^{D-1}, \qquad K_\ell^+ := B_2\bigl(\zeta_\ell,(1+\gamma)\rho_\ell\bigr) \cap\mathbb S^{D-1}.
\]
Then
\[
    W_k(X,\gamma)\subseteq\bigcup_{\ell=1}^m K_\ell \quad\Longrightarrow\quad X \subseteq \bigcup_{\ell=1}^m K_\ell^+.
\]
The construction takes $O\bigl(|X|\,kD\,T_{k,D}(\gamma)\bigr)$ time under the computational model stated in the preliminaries.
\end{lemma}

\begin{proof}
Regard $X$ as a subset of the ambient Euclidean space $\mathbb R^D$ and apply Lemma~\ref{lem:euclidean-multiplicative-certificate}. If the caps $K_1,\ldots,K_m$ cover the selected subset, the Euclidean balls $B_2(\zeta_\ell,\rho_\ell)$ also cover it. The lemma implies
\[
    X\subseteq \bigcup_{\ell=1}^m B_2\bigl(\zeta_\ell,(1+\gamma)\rho_\ell\bigr).
\]
Intersecting with $\mathbb S^{D-1}$ proves the cap-covering implication. The same construction is nonempty for nonempty $X$ and gives the stated size and running-time bounds.
\end{proof}

The set $W_k(X,\gamma)$ is chosen independently of the covering caps and works simultaneously for all their directions and radii. The remaining issue is to convert its multiplicative coverage guarantee into an additive guarantee for hyperbolic distances.

\subsection{A Radius-Independent Shell Reduction}\label{subsec:shell-reduction}
The hyperbolic ingredient is the growth of chordal cap radii under an additive increase in the distance threshold. For $c=(s,\zeta)$ with $s>0$ and a nonempty cap at level $t$, Lemma~\ref{lem:cap-growth} in Appendix~\ref{app:cap-growth} gives
\begin{equation}
    \rho_c(t+u;\bar r) \ge e^{u/2}\rho_c(t;\bar r) \qquad\text{for every }u\ge 0.
    \label{eq:large-cap-growth}
\end{equation}
For $0<\gamma\le 1$, put $u_\gamma:=2\log(1+\gamma)$. Equations~\eqref{eq:large-cap-representation}
and~\eqref{eq:large-cap-growth} imply 
\begin{equation}
    B_2\bigl(\zeta,(1+\gamma)\rho_c(t;\bar r)\bigr) \cap\mathbb S^{D-1} \subseteq K_{\bar r} (c,t+u_\gamma).
    \label{eq:large-expanded-cap}
\end{equation}
The required increase $u_\gamma$ is independent of $\bar r$, $s$, and $t$. Combining this property with radial projection gives the following structure theorem.

\begin{theorem}\label{thm:shell-reduction}
Fix an anchor $a\in\Hd$, a representative radius $\bar r>0$, and $\delta\ge 0$. Let $B\subseteq\Hd\setminus\{a\}$ be finite and nonempty, and suppose that
\[
    |d_H(a,p)-\bar r|\le\delta \qquad\text{for every }p\in B.
\]
Let $X_B\subseteq\mathbb S^{D-1}$ be the directions of $B$ in polar coordinates centered at $a$. For $0 < \gamma \le 1$, let a nonempty subset $\overline W\subseteq X_B$ satisfy the cap-covering implication of Lemma~\ref{lem:algorithmic-cap-certificate} with accuracy $\gamma$. For each direction in $\overline W$, retain one original point of $B$ having that direction, and denote the resulting subset by $W$. Then, for every $C \subseteq \Hd$ with $1\le |C|\le k$,
\begin{equation}
    \cost_W(C) \le \cost_B(C) \le \cost_W(C)+2\delta+2\log(1+\gamma).
    \label{eq:shell-reduction-error}
\end{equation}
In particular, taking $\overline W=W_k(X_B,\gamma)$ gives $|W|\le T_{k,D}(\gamma)$. Both this size bound and the additive error bound are independent of the representative radius $\bar r$.
\end{theorem}

\begin{proof}
Apply the same hyperbolic isometry to $B$, $W$, and $C$ so that $a=O$, and use the same notation for the transformed sets. For $p=(r(p),\xi(p))\in B$, put
\[
    \widehat p:=(\bar r,\xi(p)), \qquad \widehat B:=\{\widehat p:p\in B\}, \qquad \widehat W:=\{\widehat q:q\in W\}.
\]
Every point moves by at most $\delta$.

Fix $C$ and let $t:=\cost_{\widehat W}(C)$. The selected directions satisfy
\[
    \overline W\subseteq\bigcup_{c\in C}K_{\bar r}(c,t).
\]
If one of these caps is the whole sphere, then $\cost_{\widehat B}(C)\le t$. Otherwise, discard the empty caps. At least one cap remains because $\overline W$ is nonempty. Each remaining cap has a center with $s>0$ and has the form in~\eqref{eq:large-cap-representation}. The covering property of $\overline W$ implies that enlarging the chordal radius of each remaining cap by $1+\gamma$ covers $X_B$. By~\eqref{eq:large-expanded-cap}, these enlarged caps are contained in the sublevel caps at threshold $t+u_\gamma$. Thus, in all cases,
\[
    \cost_{\widehat B}(C) \le \cost_{\widehat W}(C)+u_\gamma.
\]
Using the radial-projection error twice gives
\[
    \cost_B(C) \le \cost_{\widehat B}(C)+\delta \le \cost_{\widehat W}(C)+u_\gamma+\delta \le \cost_W(C)+2\delta+2\log(1+\gamma).
\]
The lower bound follows from $W \subseteq B$. The size bound follows by retaining one original point for each selected direction.
\end{proof}

The theorem separates the Euclidean selection step from the hyperbolic error analysis. The exact cap representation and the uniform cap-growth bound convert the certificate's coverage guarantee into a hyperbolic cost bound for every admissible center set.

\subsection{Witnesses for a Radial Shell}
Set
\[
    \gamma_G:=\min\{1,e^{\Delta/2}-1\} = \min\{1,e^{\eps R_G/16}-1\}.
\]
Then $0<\gamma_G\le 1$ and 
\begin{equation}
    2\log(1+\gamma_G)\le\Delta.
    \label{eq:large-gamma-slack}
\end{equation}
For a nonempty shell $B=B(i,j)$, put
\[
    X_B:=\{\xi_i(p):p\in B\}\subseteq\mathbb S^{D-1}
\]
and compute $\overline W_B:=W_k(X_B,\gamma_G)$. For each selected direction, retain one original point of $B$
having that direction. Denote the resulting nonempty subset by $W(B)\subseteq B$.

For implementation, process the directions as an indexed list with references to their original input points, using the list version in Appendix~\ref{app:proof-cap-certificate}. No initial removal of coincident directions is needed. Merge repeated selected input indices by marking. If $M$ distinct indices remain, then $M\le\min\{|B|,T_{k,D}(\gamma_G)\}$. Pairwise comparison removes coincident selected directions in $O(M^2D)$ time while retaining one original input reference per direction. Since $M^2\le |B|T_{k,D}(\gamma_G)$, this cost is absorbed by the following bound.

\begin{lemma}\label{lem:shell-witness}
For every nonempty shell $B$, the construction returns a nonempty set $W(B)\subseteq B$ of size at most $T_{k,D}(\gamma_G)$ such that
\[
    \max_{p\in B}d_H(p,C) \le \max_{q\in W(B)}d_H(q,C)+\frac{\eps R_G}{4}
\]
for every $C\subseteq\Hd$ with $1\le |C|\le k$. The set $W(B)$ can be computed in $O\bigl(|B|\,kD\,T_{k,D}(\gamma_G)\bigr)$ time.
\end{lemma}

\begin{proof}
Apply Theorem~\ref{thm:shell-reduction} with anchor $a_i$, representative radius $\bar r_j$, $\delta=\Delta/2$, and $\gamma=\gamma_G$. The radial condition follows from the definition of $B(i,j)$. By~\eqref{eq:large-gamma-slack}, the error is at most 
\[
    2\delta+2\log(1+\gamma_G) \le \Delta+\Delta = 2\Delta = \frac{\eps R_G}{4}.
\]
The size and running-time bounds follow from Lemma~\ref{lem:algorithmic-cap-certificate} and the indexed-list implementation described above.
\end{proof}

This guarantee concerns the maximum distance over the shell. It does not require every input point to have a nearby retained representative. The set $W(B)$ is chosen independently of $C$, although the point attaining its maximum distance may depend on $C$.

\subsection{The Large-Scale Coreset}
Retain all anchors and all shell witnesses, and set
\[
    P_\eps := A_k \cup \bigcup_{\substack{1\le i\le k,\ 1\le j\le L\\ B(i,j)\ne\varnothing}} W\bigl(B(i,j)\bigr).
\]
Since $R_G\ge R_0$ and $e^x-1\ge x$ for $x\ge 0$,
\[
    \gamma_G \ge \min\left\{1,\frac{\eps R_G}{16}\right\} \ge c_0\eps, \qquad c_0:=\min\left\{1,\frac{R_0}{16}\right\}>0.
\]
Consequently,
\[
    T_{k,D}(\gamma_G) \le k!\left(\frac{A}{c_0\eps}\right)^{kD}.
\]

\begin{theorem}[Large-scale coreset] \label{thm:large}
Assume that $R_G\ge R_0$. The construction returns $A_k\subseteq P_\eps\subseteq P$ such that
\[
    \cost_{P_\eps}(C) \le \cost_P(C) \le \cost_{P_\eps}(C)+\frac{\eps R_G}{4}
\]
for every $C\subseteq\Hd$ with $1\le |C|\le k$. Consequently, every exact optimal continuous $k$-center solution for $P_\eps$ is a $(1+\eps)$-approximation for $P$. Moreover,
\[
    |P_\eps| \le k+k\left\lceil\frac{8}{\eps}\right\rceil k!\left(\frac{A}{c_0\eps}\right)^{kD} = \left(\frac{k}{\eps}\right)^{O(kD)},
\]
and the construction time is
\[
    O\left(nkD\left(\frac{k}{\eps}\right)^{O(kD)}\right).
\]
The size bound is independent of $n$ and $R_G$.
\end{theorem}

To obtain the uniform additive bound, apply Lemma~\ref{lem:shell-witness} to the shell containing each non-anchor point and use that all anchors are retained. Taking a maximum over the input points does not add the errors from different shells. For any exact optimum $C_{P_\eps}$ of the reduced instance, $P_\eps\subseteq P$ and Lemma~\ref{lem:gonzalez} give
\[
    \cost_P(C_{P_\eps}) \le \cost_{P_\eps}(C_{P_\eps})+\frac{\eps R_G}{4} = \OPT_k(P_\eps)+\frac{\eps R_G}{4} \le \OPT_k(P)+\eps\OPT_k(P) = (1+\eps)\OPT_k(P).
\]
There are at most $kL$ nonempty shells, which gives the stated size bound. Since $\sum_B|B|\le n$, summing the shell construction times gives $O\bigl(nkD\,T_{k,D}(\gamma_G)\bigr)$. Including the anchor computation and shell-list initialization, the total construction time is
\[
    O\bigl(nkD+kL+nkD\,T_{k,D}(\gamma_G)\bigr) = O\left(nkD\left(\frac{k}{\eps}\right)^{O(kD)}\right).
\]
The detailed size and construction-time analysis is given in Appendix~\ref{app:nonbounded-coreset}.


\section{Combining the Two Scale Regimes}\label{sec:algorithm}
Write $n:=|P|$. Algorithm~\ref{alg:coreset} combines the constructions of Sections~\ref{sec:bounded} and~\ref{sec:radial-shells} using the fixed threshold $R_0$. It returns $P$ when $k\ge n$. Otherwise, it computes the farthest-first scale $R_G$ and uses the bounded-scale construction when $R_G\le R_0$ and the large-scale construction when $R_G>R_0$.

\begin{algorithm}[tb]\caption{Hyperbolic continuous $k$-center coreset}\label{alg:coreset}
\begin{algorithmic}[1]
\Require A finite nonempty set $P\subseteq\BD$, an integer $k\ge 1$, and $0<\eps<1$
\Ensure A nonempty $\eps$-coreset $P_\eps\subseteq P$ 
\State $n\gets |P|$
\If{$k\ge n$}
    \State \Return $P$
\EndIf
\State Compute farthest-first anchors $A_k$ and $R_G=\cost_P(A_k)$
\State Assign each input point to a nearest anchor to obtain $P_1,\ldots,P_k$
\If{$R_G\le R_0$}
    \State Construct $P_\eps$ by the local-grid method of Section~\ref{sec:bounded} with accuracy $\eps$
\Else
    \State $\Delta\gets\eps R_G/8$, $L\gets\lceil8/\eps\rceil$, $\gamma_G\gets\min\{1,e^{\Delta/2}-1\}$
    \State Partition each $P_i\setminus A_k$ into $L$ radial shells with midpoint representative radii, as in Section~\ref{sec:radial-shells}
    \State $P_\eps\gets A_k$
    \ForAll{nonempty radial shells $B$}
        \State Form the direction list $X_B$ of the radial projections, retaining original input indices
        \State Compute $\overline W_B\gets W_k(X_B,\gamma_G)$
        \State Retain one original point of $B$ per selected direction to obtain $W(B)$
        \State $P_\eps\gets P_\eps\cup W(B)$
    \EndFor
\EndIf
\State \Return $P_\eps$
\end{algorithmic}
\end{algorithm}

If $k\ge n$, the algorithm returns the exact coreset $P$, whose size is at most $k$. Now suppose that $k<n$. The two branches are exhaustive and disjoint. Theorem~\ref{thm:bounded-scale} applies when $R_G\le R_0$, whereas Theorem~\ref{thm:large} applies when $R_G>R_0$. In either case, the output satisfies
\begin{equation}
    \cost_{P_\eps}(C) \le\cost_P(C) \le\cost_{P_\eps}(C)+\frac{\eps R_G}{4}
    \label{eq:combined-uniform-additive}
\end{equation}
for every $C\subseteq\Hd$ with $1\le |C|\le k$.

Let $C_{P_\eps}$ be any exact optimal continuous $k$-center solution for $P_\eps$. Since $P_\eps\subseteq P$, we have $\OPT_k(P_\eps)\le\OPT_k(P)$. Using~\eqref{eq:combined-uniform-additive} and $R_G\le4\OPT_k(P)$ from Lemma~\ref{lem:gonzalez}, we obtain
\[
    \cost_P(C_{P_\eps}) \le\OPT_k(P_\eps)+\frac{\eps R_G}{4} \le\OPT_k(P)+\eps\OPT_k(P) =(1+\eps)\OPT_k(P).
\]
Thus every exact optimum of the reduced instance transfers to a $(1+\eps)$-approximation for $P$.

The bounded-scale size bound $O(k\eps^{-D})$ is absorbed by the large-scale bound, so in all cases
\[
    |P_\eps|\le\left(\frac{k}{\eps}\right)^{O(kD)}.
\]
The construction-time bounds are likewise bounded by
\[
    O\left(nkD\left(\frac{k}{\eps}\right)^{O(kD)}\right).
\]
This proves Theorem~\ref{thm:main}.

\subsection{A strong coreset guarantee}\label{subsec:strong-coreset}
The uniform additive estimate also yields a multiplicative guarantee for every admissible center set after adjusting the accuracy parameter. A nonempty subset $Q\subseteq P$ is a \emph{strong $\eps$-coreset} if
\[
    \cost_Q(C)\le\cost_P(C)\le(1+\eps)\cost_Q(C)
\]
for every $C\subseteq\Hd$ with $1\le |C|\le k$.

\begin{corollary}\label{cor:strong-coreset}
Under the assumptions of Theorem~\ref{thm:main}, one can construct a strong $\eps$-coreset of size $(k/\eps)^{O(kD)}$ in
\[
    O\left(nkD\left(\frac{k}{\eps}\right)^{O(kD)}\right)
\]
time.
\end{corollary}

\begin{proof}
Run Algorithm~\ref{alg:coreset} with accuracy
\[
    \eta:=\frac{\eps}{1+\eps},
\]
and denote its output by $Q$. If $k\ge n$, then $Q=P$ and the guarantee is exact. Assume henceforth that $k<n$, and let $R_G$ be the farthest-first scale computed in this run.

Fix any $C\subseteq\Hd$ with $1\le |C|\le k$. By~\eqref{eq:combined-uniform-additive}, applied with accuracy
$\eta$,
\[
    \cost_P(C)\le\cost_Q(C)+\frac{\eta R_G}{4}.
\]
Lemma~\ref{lem:gonzalez} and the definition of $\OPT_k(P)$ give
\[
    \frac{\eta R_G}{4} \le\eta\OPT_k(P) \le\eta\cost_P(C).
\]
Consequently,
\[
    (1-\eta)\cost_P(C)\le\cost_Q(C).
\]
Since $1-\eta=1/(1+\eps)>0$,
\[
    \cost_P(C) \le\frac{\cost_Q(C)}{1-\eta} =(1+\eps)\cost_Q(C).
\]
The inequality $\cost_Q(C)\le\cost_P(C)$ follows from $Q\subseteq P$. The construction is independent of $C$, so these inequalities hold simultaneously for all admissible center sets.

Finally, $\eps/2\le\eta\le\eps$ because $0<\eps<1$. Replacing $\eps$ by $\eta$ therefore changes neither the stated asymptotic size bound nor the construction-time bound.
\end{proof}

Every strong $\eps$-coreset also satisfies the weak solution-transfer guarantee. Indeed, if $C_Q$ is any exact optimum for $Q$, then
\[
    \cost_P(C_Q) \le(1+\eps)\OPT_k(Q) \le(1+\eps)\OPT_k(P).
\]
Hence the coreset-size lower bounds in Section~\ref{sec:coreset-lower-bound}, which hold even for the weak solution-transfer guarantee, also apply to strong coresets.


\section{Lower Bounds for Weak Coresets}\label{sec:coreset-lower-bound}
We establish size lower bounds for the weak solution-transfer notion used in this paper. Throughout this section, coresets are nonempty unweighted subsets of the input, and centers may lie anywhere in $\Hd$. To show that a subset $Q\subseteq P$ is not a weak $\eps$-coreset, it suffices to exhibit a center set $C_Q$ with $1\le |C_Q|\le k$
such that
\begin{equation}
    \cost_Q(C_Q)=\OPT_k(Q), \qquad \cost_P(C_Q)>(1+\eps)\OPT_k(P).
    \label{eq:weak-lower-bound-obstruction}
\end{equation}
An arbitrary query with a large cost discrepancy would not suffice, since the weak guarantee concerns exact optima of the reduced instance. The constructions below use distant components and auxiliary anchor pairs to enforce this exact optimality. These anchor pairs are part of the lower-bound instances and are unrelated to the farthest-first anchors used by the upper-bound algorithm.

\subsection{An angular lower bound}
The first construction applies for every $k\ge2$ and uses angular separation on a fixed-radius hyperbolic sphere.

\begin{theorem}[Coreset-size lower bound]\label{thm:angular-lower-bound}
For every fixed integer $D\ge1$, there exist constants $\eps_0>0$ and $c_D>0$ such that the following holds. For every integer $k\ge2$ and every $0<\eps\le\eps_0$, there is a finite set $P\subseteq\Hd$ for which every nonempty weak $\eps$-coreset $Q\subseteq P$ equals $P$ and satisfies
\[
    |Q|=|P|\ge c_D k\eps^{-(D-1)/2}.
\]
\end{theorem}

For $D\ge2$, choose an $\alpha$-separated set $U\subseteq\mathbb S^{D-1}$ with $\alpha=\Theta(\sqrt\eps)$ and $|U|=\Omega_D(\eps^{-(D-1)/2})$. The constant in the separation parameter is chosen sufficiently large. For $D=1$, use $U=\mathbb S^0=\{-1,+1\}$; the resulting bound is $\Omega(k)$. In polar coordinates around a base point $o$, define
\[
    x_u:=\left(\frac12,u\right), \qquad c_u:=\left(\frac12+2\eps,-u\right), \qquad u\in U.
\]
The hyperbolic law of cosines gives
\[
    \cosh d_H(x_v,c_u) = \cosh(1+2\eps)-\sinh(1/2)\sinh(1/2+2\eps) \bigl(1-\langle u,v\rangle\bigr).
\]
Consequently,
\[
    d_H(x_u,c_u)=1+2\eps, \qquad d_H(x_v,c_u)\le1\quad(v\ne u).
\]
Thus $c_u$ covers all of $X:=\{x_v:v\in U\}$ except $x_u$ at radius $1$, while the base point $o$ covers $X$ at radius $1/2$. The query center $c_u$ need not be an input point.

Take $k-1$ mutually distant isometric copies $X_1,\ldots,X_{k-1}$ with base points $o_1,\ldots,o_{k-1}$. Add two auxiliary anchor points $b^-,b^+$ at distance $2$ from each other, with midpoint $o_0$, far from all copies. The base points $o_0,\ldots,o_{k-1}$ cover the resulting input at radius $1$. The two anchors together with one point from each copy form $k+1$ points at pairwise distances at least $2$, so $\OPT_k(P)=1$.

Appendix~\ref{app:coreset-lower-bound} shows that every weak coreset $Q$ must retain both anchors and meet every copy. The same $k+1$-point packing then forces $\OPT_k(Q)=1$. If $x_{i,u}\in X_i$ is omitted, consider
\[
    C_{i,u}:=\{o_0,c_{i,u}\} \cup\{o_j:1\le j\le k-1,\ j\ne i\}.
\]
This set has $k$ centers and covers $Q$ at radius $1$, so it is an exact optimum of $Q$. The separation ensures that no center from another copy is closer to $x_{i,u}$ than $c_{i,u}$. Hence
\[
    \cost_P(C_{i,u}) \ge d_H(x_{i,u},C_{i,u}) =1+2\eps>(1+\eps)\OPT_k(P).
\]
This contradicts the weak coreset guarantee. Thus every input point must be retained. The full separation and anchor-forcing arguments are given in Appendix~\ref{app:coreset-lower-bound}.

\subsection{A local-dimensional lower bound}\label{subsec:local-dimensional-lower}
When the center budget is sufficiently large relative to $D$, a grid construction gives a stronger dependence on the accuracy parameter for unrestricted inputs. The basic component is realized in a small hyperbolic region; the full instance combines distant copies of this component with auxiliary anchor pairs.

\begin{theorem}[Local-dimensional lower bound]\label{thm:local-dimensional-lower}
Let $D\ge1$ and $k\ge4D$ be integers, and let $0<\eps\le1/128D$. There is a finite set $P\subseteq\Hd$ such that every nonempty weak $\eps$-coreset $Q\subseteq P$ equals $P$, and
\[
    |P|\ge(64D)^{-D}(k-4D+1)\eps^{-D}.
\]
In particular, for $k\ge8D$, the required coreset size is $\Omega_D(k\eps^{-D})$.
\end{theorem}

The proof in Appendix~\ref{app:local-dimensional-lower} starts with a Euclidean grid of at least $(64D\eps)^{-D}$ points. For each grid point, $2D$ centers cover all other grid points while leaving that point outside the corresponding enlarged balls. The construction has a strict distance margin that survives transfer to a sufficiently small Poincar\'e ball.
The resulting component $X\subseteq B_H(o,R)$ has the following property: for every $p\in X$, there is a set $\mathcal C_p$ of $m:=2D$ centers such that
\[
    \cost_{X\setminus\{p\}}(\mathcal C_p)<R, \qquad d_H(p,\mathcal C_p)\ge(1+2\eps)R.
\]

Put $g:=k-4D+1\ge1$. Take $g$ distant copies of $X$ and $m=2D$ distant anchor pairs, each pair having endpoint distance $2R$. The component base points and anchor midpoints cover the input at radius $R$. The $2m$ anchor endpoints together with one point from each component form $2m+g=k+1$ points at pairwise distances at least $2R$, so $\OPT_k(P)=R$. The appendix shows that every weak coreset must retain all anchor endpoints and meet every component, forcing $\OPT_k(Q)=R$ as well.

If a point in one component is omitted, use its $m$ exposing centers, the $g-1$ base points of the other components, and the $m$ anchor midpoints. The total number of centers is
\[
    m+(g-1)+m=k.
\]
This center set is an exact optimum of $Q$ at radius $R$ but leaves the omitted point farther than $(1+\eps)R$. Thus $Q=P$, and counting the component points gives the stated bound. Although the components are placed far apart, this construction does not satisfy the partition promise of Section~\ref{sec:sep-inputs} with exactly $k$ nonempty separated blocks.

\subsection{A joint-parameter lower bound}\label{subsec:joint-parameter-lower}
The preceding constructions concern the dependence on the accuracy parameter. We next rule out a coreset-size bound that is polynomial jointly in $k$ and $1/\eps$, already in $\mathbb H^2$. We combine the tree construction of
Kluk et al.~\cite[Theorem~1.12]{KlukEtAl2026} with the low-distortion hyperbolic embedding of Sarkar~\cite[Theorem~6]{Sarkar2011}. The embedding preserves the distance gaps needed to expose individual input points. Auxiliary anchor pairs then make the exposing queries exact optima of the reduced instance. The latter step is necessary for the weak solution-transfer guarantee.

\begin{theorem}[Joint-parameter lower bound]\label{thm:joint-parameter-lower}
Let $D\ge2$, let $h\ge1$ and $k\ge2h$ be integers, and let $0<\eps\le1/(h+2)$. There is a finite set $P\subseteq\Hd$ with
\[
    |P|=(k-2h+1)2^h+2h
\]
such that every nonempty weak $\eps$-coreset $Q\subseteq P$ equals $P$. The points can be placed in a totally geodesic copy of $\mathbb H^2$, while the centers of every reduced instance remain unrestricted in $\Hd$.
\end{theorem}

The tree component has $2^h$ input leaves. For every leaf, $h$ auxiliary query vertices cover all other leaves at tree distance at most $h$, whereas the omitted leaf has distance $h+2$ from the query set. Embedding the finite tree, including these query vertices, with sufficiently small distortion and a common distance scale yields a component $X\subseteq B_H(o,R)$ with $|X|=2^h$. For every $p\in X$, there is a set $\mathcal C_p$ of $h$ centers such that
\[
    \cost_{X\setminus\{p\}}(\mathcal C_p)\le R, \qquad d_H(p,\mathcal C_p)>(1+\eps)R.
\]
The distortion and strict-gap calculation are given in Appendix~\ref{app:joint-parameter-lower}.

Take $g:=k-2h+1$ distant copies of this component and add $h$ distant anchor pairs with endpoint distance $2R$. The component base points and anchor midpoints cover $P$ at radius $R$ using $g+h\le k$ centers. The $2h$ anchor endpoints and one point from each component form $2h+g=k+1$ points at pairwise distances at least $2R$, which gives $\OPT_k(P)=R$. As shown in the appendix, every weak coreset retains both endpoints of every anchor pair and at least one point from every component. These retained points also force $\OPT_k(Q)=R$.

For a point omitted from one component, combine its $h$ exposing centers with the $g-1$ base points of the other components and all $h$ anchor midpoints. This uses
\[
    h+(g-1)+h=k
\]
centers and is an exact optimum of $Q$ at radius $R$. Separation keeps all other centers far from the omitted point, so its cost on $P$ exceeds $(1+\eps)R$. Thus every point is mandatory, and $|P|=g2^h+2h$. The argument that smaller radii are impossible uses only pairwise input distances and the triangle inequality; it does not restrict centers to the embedded tree or to the hyperbolic plane.

\begin{corollary}\label{cor:no-joint-polynomial}
For every fixed $D\ge2$, there are no constants $C_D>0$ and $a,b\ge0$, independent of $k$ and $\eps$, such that every finite nonempty input in $\Hd$ admits a weak $\eps$-coreset of size at most $C_Dk^a\eps^{-b}$ for all $k\ge1$ and $0<\eps<1$. In particular, the Euclidean size bound $O_D(k\eps^{-D})$ cannot hold for all hyperbolic inputs.
\end{corollary}

\begin{proof}
For each integer $h\ge1$, set $k=2h$ and $\eps=1/(h+2)$ in Theorem~\ref{thm:joint-parameter-lower}. Every weak coreset of the resulting instance has size
\[
    2^h+2h\ge2^h.
\]
A bound of the claimed polynomial form would instead give
\[
    C_Dk^a\eps^{-b}=C_D(2h)^a(h+2)^b.
\]
The right-hand side is polynomial in $h$ and is smaller than $2^h$ for all sufficiently large $h$, a contradiction.
\end{proof}

The corollary is a joint-parameter obstruction. For a fixed $\eps$, the condition $\eps\le1/(h+2)$ bounds $h$, so this construction does not establish an exponential dependence on $k$ at fixed accuracy. It does not exclude a bound of the form $f_D(\eps)k^{a_D}$ with $a_D$ independent of $\eps$, nor does it match the full parameter dependence of our upper bound. The local-dimensional lower bound gives an $\eps^{-D}$ obstruction when $k$ is sufficiently large relative to $D$, while the angular construction also applies for every $k\ge2$. Since every strong coreset satisfies the weak solution-transfer guarantee, all three lower bounds apply to strong coresets as well.


\section{Separated Inputs in Hyperbolic and Euclidean Spaces}\label{sec:sep-inputs}
We now consider inputs partitioned into exactly $k$ sufficiently separated nonempty clusters. For this class, hyperbolic farthest-point coresets give simultaneous multiplicative and additive weak guarantees with size independent of the input cardinality and cluster radii. Under the same separation condition, an additive guarantee admits no size bound depending only on $k$, $D$, and the additive error in the Euclidean plane, already for $k=2$.

\subsection{The separation condition}\label{subsec:sep-condition}
Let $(X,d)$ be either $\mathbb H^D$ with its hyperbolic metric or $\mathbb R^D$ with its Euclidean metric. For a finite nonempty set $A\subseteq X$, define
\[
    \operatorname{rad}(A) := \min_{c\in X}\max_{p\in A}d(p,c).
\]
For finite nonempty sets $A,B\subseteq X$, write
\[
    \operatorname{dist}(A,B) := \min_{a\in A,\,b\in B}d(a,b).
\]
Let $k\ge1$ and suppose that
\[
    P=P_1\mathbin{\dot\cup}\cdots\mathbin{\dot\cup}P_k
\]
is a partition into exactly $k$ nonempty subsets. Put
\[
    R:=\max_{1\le i\le k}\operatorname{rad}(P_i).
\]
The separation condition is
\begin{equation}
    \operatorname{dist}(P_i,P_j)>2R \qquad\text{for all distinct }i,j.
    \label{eq:sep-condition}
\end{equation}
Here $R$ is the maximum enclosing radius of an individual block, not the enclosing radius of the full input or the farthest-first scale $R_G$.

\begin{lemma}\label{lem:sep-one-center}
Assume~\eqref{eq:sep-condition}. Then $\OPT_k(P)=R$. Moreover, let $Q_i\subseteq P_i$ be nonempty for every $i$, and let $Q=\bigcup_{i=1}^k Q_i$. Every exact optimal continuous $k$-center solution $C_Q$ for $Q$ contains exactly $k$ centers that can be indexed as $c_1,\ldots,c_k$ so that
\[
    \max_{q\in Q_i}d(q,c_i) \le\OPT_k(Q)\le R \qquad\text{for every }i.
\]
\end{lemma}

\begin{proof}
The minimum-enclosing-ball centers of the blocks cover $P$ at radius $R$, so $\OPT_k(P)\le R$. Any closed ball of radius at most $R$ intersects at most one block: otherwise two points from different blocks would have distance at most $2R$, contradicting~\eqref{eq:sep-condition}. If $P$ admitted a cover by at most $k$ balls of radius $r<R$, covering the $k$ nonempty blocks would require exactly one ball per block. Each such ball would cover its entire block, implying $\operatorname{rad}(P_i)\le r$ for every $i$, a contradiction. Thus $\OPT_k(P)=R$. If $R=0$, the same equality follows directly because every block is a singleton.

Now put $r_Q:=\OPT_k(Q)\le\OPT_k(P)=R$. The closed radius-$r_Q$ balls centered at the points of $C_Q$ cover $Q$, and each intersects at most one $Q_i$. Since all $k$ sets $Q_i$ are nonempty and $|C_Q|\le k$, there must be exactly one ball for each $Q_i$. No other ball can cover a point of that $Q_i$, so its associated ball covers all of $Q_i$.
Indexing the centers accordingly proves the claim. 
\end{proof}

The center $c_i$ need not minimize the farthest distance within $Q_i$; its cost is only required to be at most the common optimum $\OPT_k(Q)$. This is why the construction below uses a farthest-point guarantee valid for every possible center.

For $\alpha>0$ and $0<\eps<1$, a nonempty subset $Q\subseteq P$ provides \emph{simultaneous guarantees} if every exact optimal continuous $k$-center solution $C_Q$ for $Q$ satisfies
\begin{equation}
    \cost_P(C_Q) \le\min\bigl\{\OPT_k(P)+\alpha, (1+\eps)\OPT_k(P)\bigr\}.
    \label{eq:sep-simultaneous}
\end{equation}
If only $\cost_P(C_Q)\le\OPT_k(P)+\alpha$ is required for every such $C_Q$, we call $Q$ an \emph{additive $\alpha$ weak coreset}.

\subsection{Simultaneous guarantees in hyperbolic space}\label{subsec:sep-hyperbolic}
For a finite nonempty set $A\subseteq\Hd$, write
\[
    F_A(c):=\max_{p\in A}d_H(p,c).
\]
The hyperbolic farthest-point coreset theorem of~\cite{Park2026Boundary} gives, for fixed $D\ge2$ and $0<\eta<1$, a nonempty subset $A_\eta\subseteq A$ of size $O_D(\eta^{-(D-1)/2})$ satisfying, simultaneously for every $c\in\Hd$,
\begin{equation}
    F_A(c)\le F_{A_\eta}(c)+\eta, \qquad F_A(c)\le\frac{F_{A_\eta}(c)}{1-\eta}.
    \label{eq:sep-farthest-guarantees}
\end{equation}
We apply this result independently to the separated blocks.

\begin{theorem}\label{thm:sep-hyperbolic}
Let $D\ge2$ be fixed, and suppose that a finite nonempty set $P\subseteq\Hd$ admits a partition into exactly $k\ge1$ nonempty blocks satisfying~\eqref{eq:sep-condition}. For $\alpha>0$ and $0<\eps<1$, put
\[
    \eta:=\min\left\{\alpha,\frac{\eps}{1+\eps}\right\}.
\]
One can construct a nonempty subset $Q\subseteq P$ of size
\[
    |Q|=O_D\!\left(k\eta^{-(D-1)/2}\right)
\]
that provides the simultaneous guarantees in~\eqref{eq:sep-simultaneous}. The size bound is independent of $|P|$, $R$, and the distances between the blocks.
\end{theorem}

\begin{proof}
The parameter satisfies $0<\eta<1$. For every $i$, select a nonempty farthest-point coreset $Q_i\subseteq P_i$ with accuracy $\eta$, and set $Q:=\bigcup_{i=1}^k Q_i$. Summing the per-block size bounds gives the stated size.

Let $C_Q$ be any exact optimum for $Q$, and put $r_Q:=\OPT_k(Q)$. By Lemma~\ref{lem:sep-one-center}, the centers can be indexed as $c_1,\ldots,c_k$ with $F_{Q_i}(c_i)\le r_Q\le R$. For every $p\in P_i$, $d_H(p,C_Q)\le d_H(p,c_i)$. Thus~\eqref{eq:sep-farthest-guarantees} gives both
\[
    \cost_P(C_Q) \le\max_i F_{P_i}(c_i) \le r_Q+\eta
\]
and
\[
    \cost_P(C_Q) \le\max_i F_{P_i}(c_i) \le\frac{r_Q}{1-\eta}.
\]
Since $r_Q\le R=\OPT_k(P)$, $\eta\le\alpha$, and $(1-\eta)^{-1}\le1+\eps$, these inequalities imply \eqref{eq:sep-simultaneous}. The argument holds for every exact optimal solution $C_Q$.
\end{proof}

\begin{corollary}\label{cor:sep-tight}
For every fixed $D\ge2$, every $k\ge2$, and all sufficiently small $\eps>0$, the optimal worst-case size of a weak coreset satisfying~\eqref{eq:sep-simultaneous} with $\alpha=\eps$, over inputs satisfying~\eqref{eq:sep-condition}, is
\[
    \Theta_D \left(k\eps^{-(D-1)/2}\right).
\]
\end{corollary}

\begin{proof}
When $\alpha=\eps$, we have $\eta=\eps/(1+\eps)=\Theta(\eps)$, so the upper bound follows from Theorem~\ref{thm:sep-hyperbolic}. For the lower bound, use the instances of Theorem~\ref{thm:angular-lower-bound}. Their $k-1$ witness families lie in radius-$1/2$ balls about their base points, and the remaining block is an anchor pair of enclosing radius $1$. Distinct base points have distance at least $20$. Points in two witness families are therefore at distance at least $19$, and points in a witness family and the anchor pair are at distance at least $20-1/2-1=37/2$. Hence these exactly $k$ nonempty blocks satisfy \eqref{eq:sep-condition} with $R=1$. The multiplicative weak coreset lower bound $\Omega_D(k\eps^{-(D-1)/2})$ applies in particular when both guarantees are required.
\end{proof}

\paragraph{Recovering the partition.}
Under the separation promise, the partition need not be supplied. Run the first $k$ steps of Gonzalez's farthest-first traversal and then assign every point to a nearest selected anchor. Each block has diameter at most $2R$. Once a block contains a selected anchor, all its points have distance at most $2R$ from the selected anchor set. Every point of an unrepresented block has distance greater than $2R$ from every selected anchor. Therefore, the first $k$ selected anchors lie in distinct blocks. After these selections, every point has distance at most $2R$ from its own block's anchor and greater than $2R$ from all other anchors, so nearest-anchor assignment recovers the partition. For fixed dimension, this takes $O(|P|k)$ time under the stated computational model. The farthest-point coreset construction is then applied separately to each recovered block. Neither $R$ nor the partition needs to be known in advance. This is a recovery procedure under the promise, not a test that an arbitrary input satisfies it.

\subsection{An additive lower bound in the Euclidean plane}\label{subsec:sep-euclidean}
The one center per cluster property also holds in Euclidean space. However, the same separation condition does not yield an additive weak coreset size bound independent of both the input cardinality and its geometric scale.

\begin{theorem}\label{thm:sep-euclidean}
For every $\alpha>0$ and every integer $N\ge3$, there is a finite set $P\subseteq\mathbb R^2$ with $|P|=N+2$ that admits a partition into exactly two nonempty blocks satisfying \eqref{eq:sep-condition} such that every nonempty additive $\alpha$ weak coreset $Q\subseteq P$ for continuous $2$-center equals $P$. Consequently, no coreset-size bound depending only on $k$, $D$, and $\alpha$ is possible for this guarantee, even under~\eqref{eq:sep-condition}.
\end{theorem}

The construction uses a regular polygon and a distant auxiliary anchor pair. Let $u_1,\ldots,u_N\in\mathbb S^1$ be equally spaced, put
\[
    \delta_N:=1-\cos(2\pi/N), \qquad R\ge\frac{16\alpha}{\delta_N},
\]
and define
\[
    S:=\{p_i=(R/2)u_i:1\le i\le N\}, \qquad A:=\{9Re_1,11Re_1\}, \qquad P:=S\cup A,
\]
where $e_1=(1,0)$. Then
\[
    \operatorname{rad}(S)=R/2, \qquad \operatorname{rad}(A)=R, \qquad \operatorname{dist}(S,A)\ge17R/2>2R.
\]
Thus $S,A$ form the required partition, and Lemma~\ref{lem:sep-one-center} gives $\OPT_2(P)=R$. Also, $R\ge8\alpha$ because $\delta_N\le2$.

Suppose a subset $Q$ retains both anchors and at least one polygon vertex but omits $p_i$. Let
\[
    c_i:=-(R/2+2\alpha)u_i, \qquad m:=10Re_1, \qquad C_i:=\{c_i,m\}.
\]
For $j\ne i$, we have $1-\langle u_i,u_j\rangle\ge\delta_N$, and hence
\[
    \|p_j-c_i\|_2^2 =(R+2\alpha)^2 -\left(\frac{R^2}{2}+2\alpha R\right)(1-\langle u_i,u_j\rangle) \le R^2.
\]
The last inequality follows from $R\delta_N\ge16\alpha$ and $R\ge8\alpha$. Thus $c_i$ covers every other polygon vertex at radius $R$, whereas $\|p_i-c_i\|_2=R+2\alpha$. The anchor midpoint $m$ covers $A$ at radius $R$ and satisfies
\[
    \|p_i-m\|_2\ge19R/2>R+2\alpha.
\]
The two retained anchors together with any retained polygon vertex form three points at pairwise distances at least $2R$. No two balls of radius less than $R$ can cover them, so $\OPT_2(Q)=R$. Consequently, $C_i$ is an exact optimum for $Q$, but
\[
    \cost_P(C_i) \ge d(p_i,C_i) =R+2\alpha >\OPT_2(P)+\alpha.
\]
Appendix~\ref{app:sep-euclidean} also handles subsets that omit an anchor or retain no polygon vertex, proving that every proper nonempty subset fails.

For a quantitative bound with a prescribed ratio $R/\alpha\ge18$, choose $N=\lfloor\sqrt{R/(2\alpha)}\rfloor\ge3$.
Since
\[
    \delta_N=2\sin^2(\pi/N)\ge\frac8{N^2}, \qquad \frac{16\alpha}{\delta_N}\le2\alpha N^2\le R,
\]
the same construction has $\OPT_2(P)=R$ and requires $N+2=\Omega(\sqrt{R/\alpha})$ points in every additive $\alpha$ weak coreset.


\section{Conclusion}\label{sec:conclusion}
We constructed coresets for continuous $k$-center in fixed-dimensional hyperbolic space with size $(k/\eps)^{O(kD)}$ and construction time $O\bigl(nkD(k/\eps)^{O(kD)}\bigr)$. The size bound is independent of $n$, and neither bound depends on the enclosing radius under the stated computational model. For fixed $D$, $k$, and $\eps$, the construction therefore takes linear time and returns a constant-size subset. After adjusting the accuracy parameter, the same construction yields a strong $\eps$-coreset with the same asymptotic bounds, preserving the cost of every nonempty center set $C\subseteq\Hd$ with $|C|\le k$ within a factor of $1+\eps$.

The construction combines local Euclidean grids at bounded scales with radial shells at larger scales. The hyperbolic law of cosines represents distance sublevel sets on each representative sphere as unions of at most $k$ spherical caps. Our radius-independent shell reduction converts multiplicative enlargement of their chordal radii into a uniform additive hyperbolic error independent of the representative radius. This justifies applying a recursive Euclidean multiplicative ball certificate to the input directions and retaining the corresponding original input points.

Our lower bounds hold even for weak coresets that require every exact optimum of the reduced instance to be a $(1+\eps)$-approximation for the original input. The local-dimensional construction requires $\Omega_D(k\eps^{-D})$ points for $k\ge8D$ and sufficiently small $\eps$. The joint-parameter construction requires $2^{\Omega(k)}$ points on instances with $\eps=\Theta(1/k)$, already in $\mathbb H^2$. Consequently, no coreset-size bound polynomial jointly in $k$ and $1/\eps$ is possible for all hyperbolic inputs. In particular, the Euclidean bound $O_D(k\eps^{-D})$ cannot extend to this setting.

For inputs partitioned into exactly $k$ nonempty clusters satisfying \eqref{eq:sep-condition}, we obtain simultaneous additive and multiplicative weak guarantees by applying farthest-point coresets within the clusters. When both error parameters equal $\eps$, the optimal worst-case size is $\Theta_D(k\eps^{-(D-1)/2})$ for fixed $D\ge2$, $k\ge2$, and sufficiently small $\eps$. Under the same separation condition, no additive weak coreset-size bound depending only on $k$, $D$, and the additive error parameter is possible in the Euclidean plane, even for $k=2$. Thus the separation promise permits radius-independent additive guarantees in hyperbolic space that are unavailable in the Euclidean plane under the same promise.

Several questions remain. The joint-parameter lower bound does not match the dependence on $k$ and $1/\eps$ in our construction, nor does it establish that exponential dependence on $k$ is necessary at fixed accuracy. For fixed accuracy, the planar graph bounds transferred in Appendix~\ref{app:planar-transfer} imply the existence of strong coresets in $\mathbb H^2$ whose size is polynomial in $k$. However, this transfer does not directly provide an efficient continuous construction achieving those bounds. Obtaining such a construction and determining the best dependence on $k$ in higher dimensions are natural directions. Further questions concern efficient coreset constructions for discrete $k$-center, where the reduced instance must choose its centers from the retained input points, and extensions of the radial-shell method to other hyperbolic clustering objectives.

\bibliographystyle{plainurl}
\bibliography{hyper_k-center_coreset}

\newpage
\appendix

\section{Proof of Lemma~\ref{lem:gonzalez}}\label{app:proof-gonzalez}

\begin{proof}
If $k\ge n$, then $A_k=P$ and $R_G=\OPT_k(P)=0$. The anchor set and its singleton partition can be returned in $O(n)$ time. Hence assume that $1\le k<n$.

Since $A_k\subseteq P\subseteq\Hd$ and $|A_k|=k$, the set $A_k$ is a feasible continuous center set. Therefore,
\[
    \OPT_k(P)\le\cost_P(A_k)=R_G.
\]

We prove the upper bound directly from the farthest-first selection rule. For the analysis, choose an additional point
\[
    a_{k+1}\in \arg\max_{p\in P}d_H(p,A_k).
\]
Because $P$ is a set of distinct points and $k<n$, we have $R_G>0$ and $a_{k+1}\notin A_k$. Thus $a_1,\ldots,a_{k+1}$ are distinct.

For every $1\le i<j\le k+1$, the selection rule gives
\[
    d_H(a_i,a_j) \ge d_H(a_j,A_{j-1}) = \max_{p\in P}d_H(p,A_{j-1}) \ge\max_{p\in P}d_H(p,A_k) = R_G.
\]
The first inequality uses $a_i\in A_{j-1}$, and the second inequality uses $A_{j-1}\subseteq A_k$.

Let $C^\star$ be an optimal continuous $k$-center solution for $P$, and put $r^\star:=\OPT_k(P)$. Assign each of the $k+1$ points $a_1,\ldots,a_{k+1}$ to a nearest center in $C^\star$. Since $|C^\star|\le k$, two distinct points $a_i,a_j$ are assigned to the same center $c^\star\in C^\star$. Consequently,
\[
    R_G \le d_H(a_i,a_j) \le d_H(a_i,c^\star)+d_H(c^\star,a_j) \le 2r^\star.
\]
In particular,
\[
    R_G\le 2\OPT_k(P)\le 4\OPT_k(P).
\]

Finally, we analyze the running time. For each input point $p$, maintain its current nearest-anchor distance
\[
    \delta_i(p):=d_H(p,A_i)
\]
and an index of an anchor attaining this distance. After choosing $a_1$, initialize $\delta_1(p)=d_H(p,a_1)$ for all $p\in P$. Whenever a new anchor $a_i$ is selected, update
\[
    \delta_i(p) = \min\bigl\{\delta_{i-1}(p),d_H(p,a_i)\bigr\}
\]
and the associated anchor index, using the fixed tie-breaking rule.

For fixed $D$, each update takes constant time under the stated computational model. Updating all points therefore takes $O(n)$ time per anchor. Selecting a point with maximum current nearest-anchor distance also takes $O(n)$ time per iteration. After the final update, the stored indices determine the nearest-anchor partition, and
\[
    R_G=\max_{p\in P}\delta_k(p)
\]
is obtained by one scan. The total running time is $O(kn)$.
\end{proof}


\section{Proofs for the Bounded-Scale Case}\label{app:bounded-scale}
We prove Theorem~\ref{thm:bounded-scale}. Throughout this section, $D\ge 1$ is fixed, $0<\eps<1$, and $R_G\le R_0$, where $R_0>0$ is a fixed absolute constant. Recall that $P_i$ is the cluster assigned to the anchor $a_i$ and that $\mu_{a_i}$ is a hyperbolic isometry satisfying $\mu_{a_i}(a_i)=O$.

If $R_G=0$, every input point coincides with an anchor. Hence $P=A_k$ as point sets, and we return $P_\eps=P$. All stated guarantees hold immediately, with $|P_\eps|\le k$. We therefore assume that $0<R_G\le R_0$ for the remainder of the proof.

\subparagraph{Cluster localization.}
For every $p\in P_i$, the nearest-anchor assignment gives $d_H(p,a_i)\le R_G$. Since $\mu_{a_i}$ is an isometry,
\[
    d_H\bigl(O,\mu_{a_i}(p)\bigr) = d_H\bigl(\mu_{a_i}(a_i),\mu_{a_i}(p)\bigr) = d_H(a_i,p) \le R_G.
\]
Consequently,
\[
    \mu_{a_i}(P_i) \subseteq \Ball(O,R_G) \subseteq \Ball(O,R_0).
\]

\subparagraph{Local metric comparison.}
In Poincar\'e coordinates, $\Ball(O,R_0)$ is the Euclidean ball centered at $O$ with radius
\[
    \rho_0:=\tanh(R_0/2)<1.
\]
Recall that
\[
    L_0:=\frac{2}{1-\rho_0^2},
\]
which depends only on $R_0$.

\begin{lemma}\label{lem:local-metric-comparison}
For every $x,y \in \Ball(O,R_0)$,
\[
    2\|x-y\|\le d_H(x,y)\le L_0\|x-y\|.
\]
\end{lemma}

\begin{proof}
The hyperbolic length element in the Poincar\'e ball model is
\[
    ds_H=\frac{2\|dz\|}{1-\|z\|^2}.
\]
Let $\gamma$ be a hyperbolic geodesic joining $x$ and $y$. Since $1-\|z\|^2\le 1$ throughout $\BD$,
\[
    d_H(x,y) = \int_\gamma\frac{2\|dz\|}{1-\|z\|^2} \ge 2\int_\gamma\|dz\| \ge 2\|x-y\|.
\]

For the upper bound, the Euclidean segment $\overline{xy}$ lies in the Euclidean ball of radius $\rho_0$. Hyperbolic distance is at most the hyperbolic length of this segment, so
\[
    d_H(x,y) \le \int_{\overline{xy}} \frac{2\|dz\|}{1-\|z\|^2} \le \frac{2}{1-\rho_0^2} \int_{\overline{xy}}\|dz\| = L_0\|x-y\|.
\]
\end{proof}

\subparagraph{Grid construction.}
Set
\[
    h:=\frac{\eps R_G}{4L_0\sqrt D}
\]
and let
\[
    G_h := \left\{ h\bigl(z+[0,1)^D\bigr): z\in\mathbb Z^D \right\}.
\]
These half-open cells partition $\R^D$. In particular, each point belongs to a unique cell, including points on grid boundaries.

\begin{lemma}\label{lem:grid-cell-diameter}
If $x,y \in \Ball(O,R_0)$ belong to the same cell of $G_h$, then
\[
    d_H(x,y)\le\frac{\eps R_G}{4}.
\]
\end{lemma}

\begin{proof}
Each coordinate of $x-y$ has absolute value at most $h$. Thus $\|x-y\|\le h\sqrt D$. By Lemma~\ref{lem:local-metric-comparison},
\[
    d_H(x,y) \le L_0\|x-y\| \le L_0h\sqrt D = \frac{\eps R_G}{4}.
\]
\end{proof}

For each cluster $P_i$, apply $\mu_{a_i}$ and identify the occupied cells of $G_h$. From each occupied cell, retain one original point of $P_i$ whose image lies in that cell. For the cell containing $O$, retain the anchor $a_i$. Let $P_{\eps,i}\subseteq P_i$ be the retained set and define
\[
    P_\eps:=\bigcup_{i=1}^k P_{\eps,i}.
\]
By construction, $A_k\subseteq P_\eps\subseteq P$.

\subparagraph{Correctness.}
For each $p\in P_i$, let $\pi(p)\in P_{\eps,i}$ be the representative whose image lies in the same cell as $\mu_{a_i}(p)$. Lemma~\ref{lem:grid-cell-diameter} and the isometry property give
\[
    d_H(p,\pi(p)) = d_H\bigl(\mu_{a_i}(p),\mu_{a_i}(\pi(p))\bigr) \le\frac{\eps R_G}{4}.
\]

Let $C \subseteq \Hd$ be any center set with $1\le |C|\le k$. For every $p\in P$, the triangle inequality gives
\[
    d_H(p,C) \le d_H(p,\pi(p))+d_H(\pi(p),C) \le \frac{\eps R_G}{4}+\cost_{P_\eps}(C).
\]
Taking the maximum over $p\in P$ and using $P_\eps \subseteq P$, we obtain
\begin{equation}
    \cost_{P_\eps}(C) \le \cost_P(C) \le \cost_{P_\eps}(C)+\frac{\eps R_G}{4}.
    \label{eq:bounded-grid-additive}
\end{equation}

Now let $C_{P_\eps}$ be any exact optimal continuous $k$-center solution for $P_\eps$. Since $P_\eps\subseteq P$,
\[
    \OPT_k(P_\eps)\le\OPT_k(P).
\]
Applying \eqref{eq:bounded-grid-additive} with $C=C_{P_\eps}$ and using $R_G \le 4\OPT_k(P)$ from Lemma~\ref{lem:gonzalez}, we get
\[
    \cost_P(C_{P_\eps}) \le \cost_{P_\eps}(C_{P_\eps}) + \frac{\eps R_G}{4} = \OPT_k(P_\eps)+\frac{\eps R_G}{4} \le \OPT_k(P)+\eps\OPT_k(P) = (1+\eps)\OPT_k(P).
\]
Thus every exact optimum of the reduced instance is a $(1+\eps)$-approximation for $P$.

\subparagraph{Size bound.}
Fix a cluster $P_i$. Since $\mu_{a_i}(P_i)\subseteq\Ball(O,R_G)$, its transformed points lie in the Euclidean cube
\[
    [-\rho_G,\rho_G]^D, \qquad \rho_G:=\tanh(R_G/2).
\]
Set
\[
    M:=\left\lceil\frac{\rho_G}{h}\right\rceil.
\]
Every cell containing a point of this cube has an index in $\{-M,\ldots,M\}^D$. Since at most one point is retained per occupied cell,
\[
    |P_{\eps,i}| \le (2M+1)^D \le \left(3+\frac{2\rho_G}{h}\right)^D.
\]
Using $\rho_G\le R_G/2$ and the definition of $h$,
\[
    \frac{\rho_G}{h} = \frac{4L_0\sqrt D}{\eps}\cdot\frac{\rho_G}{R_G} \le \frac{2L_0\sqrt D}{\eps}.
\]
Therefore,
\[
    |P_{\eps,i}| \le \left(3+\frac{4L_0\sqrt D}{\eps}\right)^D = O(\eps^{-D}),
\]
where the hidden constant depends only on $D$ and $R_0$. Summing over the clusters gives
\[
    |P_\eps| = \sum_{i=1}^k |P_{\eps,i}| = O(k/\eps^D).
\]

\subparagraph{Construction time.}
By Lemma~\ref{lem:gonzalez}, the anchors, the radius $R_G$, and the nearest-anchor partition can be computed in $O(kn)$ time.

For each cluster, use a direct-addressing array indexed by $\{-M,\ldots,M\}^D$. The size bound above shows that the array has $O(\eps^{-D})$ entries. Initializing these arrays over all clusters takes $O(k/\eps^D)$ time.

For a transformed point $x=\mu_{a_i}(p)=(x_1,\ldots,x_D)$, its cell index is
\[
    \left( \left\lfloor\frac{x_1}{h}\right\rfloor, \ldots, \left\lfloor\frac{x_D}{h}\right\rfloor \right).
\]
For fixed $D$, computing $x$, evaluating this index, and accessing the corresponding array entry take $O(1)$ time under the stated computational model. Initialize the zero-index cell with the original anchor $a_i$. For every other cell, store a reference to the first original input point encountered in that cell.

Since each input point is processed in exactly one cluster, the point-processing time is $O(n)$. Collecting the stored representatives takes at most $O(k/\eps^D)$ additional time. The total construction time is therefore
\[
    O\left(kn+n+\frac{k}{\eps^D}\right).
\]
This proves Theorem~\ref{thm:bounded-scale}.


\section{Proofs for the Large-Scale Construction}\label{app:nonbounded}
This appendix provides the cap-growth estimate and the Euclidean multiplicative ball certificate used in Theorem~\ref{thm:shell-reduction}. We then derive the radial-shell witness guarantee from that theorem and complete the correctness, size, and running-time analysis of the large-scale construction. All running times use the computational model stated in the preliminaries.

\subsection{Ball--sphere intersections and cap growth}\label{app:cap-growth}
Fix an anchor and a representative radius $\bar r>0$. Apply the anchor-centered isometry to both points and queried
centers, and use $O$ for the transformed anchor. For $c\in\Hd$ and $t\ge 0$, recall
\[
    K_{\bar r}(c,t) = \{\xi\in\mathbb S^{D-1}: d_H\bigl((\bar r,\xi),c\bigr)\le t\}.
\]
Put $s:=d_H(O,c)$. If $s>0$, write $c=(s,\zeta)$. The hyperbolic law of cosines gives
\[
    \cosh d_H\bigl((\bar r,\xi),c\bigr) = \cosh(\bar r-s) + \sinh\bar r\,\sinh s \bigl(1-\langle\xi,\zeta\rangle\bigr).
\]
Since $\|\xi-\zeta\|_2^2=2(1-\langle\xi,\zeta\rangle)$,
\[
    d_H\bigl((\bar r,\xi),c\bigr)\le t \quad\Longleftrightarrow\quad \|\xi-\zeta\|_2^2 \le \frac{2\bigl(\cosh t-\cosh(\bar r-s)\bigr)}{\sinh\bar r\,\sinh s}.
\]
For $t\ge|\bar r-s|$, define the nonnegative radius
\[
    \rho_c(t;\bar r)^2 := \frac{2\bigl(\cosh t-\cosh(\bar r-s)\bigr)}{\sinh\bar r\,\sinh s}.
\]
Then
\[
    K_{\bar r}(c,t) = B_2\bigl(\zeta,\rho_c(t;\bar r)\bigr) \cap\mathbb S^{D-1}.
\]
If $t<|\bar r-s|$, the cap is empty. If $\rho_c(t;\bar r)\ge 2$, it is the whole sphere. The radius $\rho_c(t;\bar r)$ is defined by the displayed formula and is not truncated at $2$. If $s=0$, then $c=O$ and $d_H\bigl((\bar r,\xi),O\bigr)=\bar r$ for every direction. In this case the cap is empty for $t<\bar r$ and is the whole sphere for $t\ge\bar r$.

\begin{lemma}\label{lem:cap-growth}
Let $\bar r>0$, let $c=(s,\zeta)$ with $s>0$, and let $t\ge|\bar r-s|$. For every $u\ge 0$,
\[
    \rho_c(t+u;\bar r) \ge e^{u/2}\rho_c(t;\bar r).
\]
\end{lemma}

\begin{proof}
Put $\beta:=\cosh(\bar r-s)$. Then $\beta\ge 1$ and $\cosh t\ge\beta$. Since $t,u\ge 0$,
\[
\begin{aligned}
    \cosh(t+u)-\beta-e^u(\cosh t-\beta) &= (e^u-1)\beta-e^{-t}\sinh u\\
    &\ge(e^u-1)-\sinh u = \cosh u-1 \ge 0.
\end{aligned}
\]
Thus
\[
    \cosh(t+u)-\beta \ge e^u(\cosh t-\beta).
\]
Multiplying by $2/(\sinh\bar r\,\sinh s)>0$ and taking nonnegative square roots proves the claim.
\end{proof}

In particular, for $0<\gamma\le 1$, setting $u=2\log(1+\gamma)$ gives
\[
    B_2\bigl(\zeta,(1+\gamma)\rho_c(t;\bar r)\bigr) \cap\mathbb S^{D-1} \subseteq K_{\bar r}\bigl(c,t+2\log(1+\gamma)\bigr).
\]
This is the cap-enlargement property used in Theorem~\ref{thm:shell-reduction}. The growth estimate also applies when the initial chordal radius is zero. Empty caps can be omitted from a covering family, and a whole-sphere cap remains the whole sphere at every larger threshold.


\subsection{Proof of the approximate $k$-cap coverage certificate}\label{app:proof-cap-certificate}
We give a net-based implementation of the recursive multiplicative coreset construction for Euclidean $k$-center described in~\cite[Theorem~6.1]{agarwal2005geometric}. The certificate allows the covering balls to have different radii and enlarges each radius separately. The proof makes the size and running-time bounds explicit and also allows indexed input lists with coincident locations.

\begin{lemma}\label{lem:euclidean-multiplicative-certificate}
Let $Y\subseteq\mathbb R^d$ be finite and nonempty, let $d,m\ge 1$ be integers, and let $0<\gamma\le 1$. One can construct a nonempty subset $Q_m(Y,\gamma)\subseteq Y$ satisfying
\[
    |Q_m(Y,\gamma)| \le m!\left(\frac{A}{\gamma}\right)^{md},
\]
where $A>1$ is a universal constant, with the following property. For every $1\le\ell\le m$, every $z_1,\ldots,z_\ell\in\mathbb R^d$, and every $r_1,\ldots,r_\ell\ge 0$,
\[
    Q_m(Y,\gamma)\subseteq \bigcup_{j=1}^{\ell}B_2(z_j,r_j) \quad\Longrightarrow\quad Y\subseteq \bigcup_{j=1}^{\ell}B_2\bigl(z_j,(1+\gamma)r_j\bigr).
\]
The construction takes
\[
    O\left( |Y|\,md\,m!\left(\frac{A}{\gamma}\right)^{md} \right)
\]
time. The same guarantees hold for a finite nonempty indexed list, where cardinalities count entries and coverage depends only on their locations.
\end{lemma}

\begin{proof}
We describe a recursive construction on a nonempty indexed list $Z$ with an integer parameter $m$. Every selected entry retains its original input index. Set unions below identify repeated indices but need not identify different entries at the same location. For analysis, put
\[
    r_m^*(Z) := \min_{\substack{F\subseteq\mathbb R^d\\1\le|F|\le m}} \max_{z\in Z}\min_{f\in F}\|z-f\|_2.
\]
The algorithm does not compute this optimum.

\subparagraph{Construction.}
If $|Z|\le m$, return all entries of $Z$. Otherwise, run Gonzalez's farthest-first traversal on $Z$, selecting at most $m$ entries and stopping early if the covering radius becomes zero. Denote the resulting radius by $\widetilde r_m(Z)$. If $\widetilde r_m(Z)=0$, return the selected entries. Their locations are distinct, there are at most $m$ of them, and every input entry coincides with one of these locations. Both stopping cases therefore preserve every covering exactly.

Suppose now that $\widetilde r_m(Z)>0$. The metric argument in the proof of Lemma~\ref{lem:gonzalez} also applies to Euclidean distances and yields
\[
    r_m^*(Z)\le\widetilde r_m(Z)\le 4r_m^*(Z).
\]
Set
\[
    \tau:=\frac{\gamma\widetilde r_m(Z)}{64}.
\]
Writing $r^*:=r_m^*(Z)$, we have 
\begin{equation}
    \frac{\gamma r^*}{64} \le\tau\le\frac{\gamma r^*}{16}.
    \label{eq:app-euclidean-net-scale}
\end{equation}
Construct a maximal $\tau$-separated subset $N\subseteq Z$. Thus distinct selected entries satisfy $\|u-v\|_2>\tau$, and every entry of $Z$ lies within distance $\tau$ of $N$. Assign each entry to a nearest point of $N$, breaking ties by a fixed rule, and let $Z_u$ be the sublist assigned to $u\in N$. These sublists partition $Z$, each contains its net entry $u$, and each has Euclidean diameter at most $2\tau$.

The set $N$ is covered by at most $m$ Euclidean balls of radius $r^*$. Within any one such ball, the open balls of radius $\tau/2$ centered at the points of $N$ are pairwise disjoint and lie in a ball of radius $r^*+\tau/2$. A volume comparison and~\eqref{eq:app-euclidean-net-scale} give
\begin{equation}
    |N| \le m\left(1+\frac{2r^*}{\tau}\right)^d \le m\left(1+\frac{128}{\gamma}\right)^d \le m\left(\frac{129}{\gamma}\right)^d.
\label{eq:app-euclidean-net-size}
\end{equation}
For $m=1$, set $Q_1(Z,\gamma):=N$. For $m\ge 2$, define
\[
    Q_m(Z,\gamma) := N\cup\bigcup_{u\in N}Q_{m-1}(Z_u,\gamma).
\]
All recursive inputs are nonempty, and the parameter decreases at each recursive call. Hence the construction terminates and returns a nonempty subset of the input indices.

\subparagraph{Coverage guarantee.}
We prove the guarantee by induction on $m$. The stopping cases are exact, so consider a call with $\widetilde r_m(Z)>0$. Let $\mathcal B=\{B_2(z_j,r_j):j=1,\ldots,\ell\}$, where $1\le\ell\le m$, cover $Q_m(Z,\gamma)$, and put
\[
    R:=\max_{1\le j\le\ell}r_j.
\]
Since $N\subseteq Q_m(Z,\gamma)$, the family covers $N$. Every entry of $Z$ is within distance $\tau$ of $N$, so
\[
    Z\subseteq\bigcup_{j=1}^{\ell}B_2(z_j,R+\tau).
\]
It follows that $r^*\le R+\tau$ and hence
\[
    R\ge r^*-\tau \ge\left(1-\frac{\gamma}{16}\right)r^* \ge\frac{15}{16}r^*.
\]
In particular,
\[
    2\tau\le\frac{\gamma r^*}{8}\le\gamma R.
\]
Fix $u\in N$ and define
\[
    J_u := \{j\in\{1,\ldots,\ell\}: B_2(z_j,r_j)\cap Z_u\ne\varnothing\}.
\]
Here intersection with a list means that some entry has its location in the ball. The set $J_u$ is nonempty because $u\in Z_u$ and the covering family contains $u$.

If $J_u=\{1,\ldots,\ell\}$, choose $j^*$ with $r_{j^*}=R$ and an entry $y_0\in B_2(z_{j^*},R)\cap Z_u$. For every $y\in Z_u$,
\[
    \|y-z_{j^*}\|_2 \le\|y-y_0\|_2+\|y_0-z_{j^*}\|_2 \le 2\tau+R \le(1+\gamma)R.
\]
Thus the enlarged largest ball covers $Z_u$.

Otherwise, necessarily $m\ge 2$, and $1\le|J_u|\le\ell-1\le m-1$. The balls indexed by $J_u$ cover $Q_{m-1}(Z_u,\gamma)$, since that set lies in $Z_u$ and is part of the certificate covered by $\mathcal B$. The induction hypothesis gives
\[
    Z_u\subseteq \bigcup_{j\in J_u}B_2\bigl(z_j,(1+\gamma)r_j\bigr).
\]
When $m=1$, we have $\ell=1$ and the unique ball contains every net point, so only the first case can occur. This establishes the induction base as well. Since the sublists $Z_u$ partition $Z$, the enlarged balls cover all entries of $Z$.

\subparagraph{Size bound.}
Put $G_d(\gamma):=(129/\gamma)^d$, and let $S_m$ be a worst-case bound on the number of distinct indices returned with parameter $m$. For $m=1$, both the stopping cases and the net construction give $S_1\le G_d(\gamma)$. For $m\ge 2$, the stopping cases return at most $m$ entries. Together with~\eqref{eq:app-euclidean-net-size}, this gives
\[
    S_m \le mG_d(\gamma)(1+S_{m-1}) \le 2mG_d(\gamma)S_{m-1}.
\]
Consequently,
\[
    S_m \le m!\,2^{m-1}G_d(\gamma)^m \le m!\left(\frac{258}{\gamma}\right)^{md}.
\]
Thus $A=258$ is sufficient.

\subparagraph{Implementation and running time.}
At a recursive call on $Z'$ with parameter $q\le m$, farthest-first traversal takes $O(|Z'|qd)$ time. A maximal $\tau$-separated subset can be computed by scanning the entries and adding an entry whenever it is farther than $\tau$ from every previously selected net point. This takes $O(|Z'|\,|N|d)$ time. A further scan assigns every entry to a nearest net point within the same bound. Thus a nonterminal call takes
\[
    O\bigl(|Z'|\,qd\,G_d(\gamma)\bigr)
\]
time, and the stopping cases satisfy this upper bound as well.

At any fixed recursion depth, the active inputs are pairwise disjoint sublists of the root input $Y$. Their total length is therefore at most $|Y|$. There are at most $m$ recursion levels, so the total time is
\[
    O\bigl(|Y|\,m^2d\,G_d(\gamma)\bigr).
\]
To implement the unions, mark selected original indices as the recursion proceeds and scan the root list once at the end.
The marking array is initialized once per root call. Intermediate certificate unions need not be materialized. This requires at most $O(|Y|m)$ additional time and is absorbed by the preceding bound. Finally,
\[
    mG_d(\gamma) \le m!\left(\frac{258}{\gamma}\right)^{md},
\]
which proves the stated running-time bound.
\end{proof}

Applying this lemma to $X\subseteq\mathbb S^{D-1}\subseteq \mathbb R^D$ with $d=D$ and $m=k$, and intersecting the enlarged balls with the sphere, proves Lemma~\ref{lem:algorithmic-cap-certificate}. This step supplies the Euclidean coverage certificate. The conversion to hyperbolic additive error is established separately by Theorem~\ref{thm:shell-reduction} in the main text.


\subsection{Details for Lemma~\ref{lem:shell-witness}}\label{app:shell-witness}
Let $B=B(i,j)$ be a nonempty radial shell, and let
\[
    X_B:=\{\xi_i(p):p\in B\}, \qquad \overline W_B:=W_k(X_B,\gamma_G).
\]
For each selected direction, retain one original point of $B$ having that direction, and denote the resulting set by $W(B)$.

The geometric guarantee follows directly from Theorem~\ref{thm:shell-reduction} with anchor $a_i$, representative radius $\bar r_j$, $\delta=\Delta/2$, and $\gamma=\gamma_G$. Indeed, $|d_H(a_i,p)-\bar r_j|\le\Delta/2$ for every $p\in B$, and $2\log(1+\gamma_G)\le\Delta$. Thus, for every $C\subseteq\Hd$ with $1\le|C|\le k$,
\[
    \cost_B(C) \le\cost_{W(B)}(C) +\Delta+2\log(1+\gamma_G) \le\cost_{W(B)}(C)+2\Delta =\cost_{W(B)}(C)+\frac{\eps R_G}{4}.
\]
For the implementation, form an indexed list of directions, each retaining a reference to its original point in $B$. Apply the indexed-list construction of Lemma~\ref{lem:euclidean-multiplicative-certificate} with $d=D$ and $m=k$. Repeated occurrences of a selected input index are merged by marking, as described above. The number $M$ of distinct selected indices satisfies
\[
    M\le\min\{|B|,T_{k,D}(\gamma_G)\}.
\]
Compare their direction vectors pairwise and keep one input reference for each distinct selected direction. This does not change any directional covering condition. The surviving references define $W(B)$ and give
\[
    1\le|W(B)|\le T_{k,D}(\gamma_G).
\]
The pairwise comparisons take $O(M^2D)$ time, and
\[
    M^2\le|B|T_{k,D}(\gamma_G).
\]
Including direction computation, certificate construction, and duplicate removal, the total time is
\[
    O\bigl(|B|\,kD\,T_{k,D}(\gamma_G)\bigr).
\]
Thus both the geometric and algorithmic claims of Lemma~\ref{lem:shell-witness} follow.


\subsection{Proof of Theorem~\ref{thm:large}}\label{app:nonbounded-coreset}
\begin{proof}
Recall that
\[
    P_\eps=A_k\cup\bigcup_B W(B),
\]
where the union ranges over the nonempty radial shells. In particular, $A_k\subseteq P_\eps\subseteq P$.

\subparagraph{Uniform additive guarantee.}
Fix any $C\subseteq\Hd$ with $1\le|C|\le k$. If $p\in A_k$, then $p\in P_\eps$, so $d_H(p,C)\le\cost_{P_\eps}(C)$. Otherwise, $p$ lies in a unique nonempty shell $B$. Lemma~\ref{lem:shell-witness}, obtained from Theorem~\ref{thm:shell-reduction}, gives
\[
    d_H(p,C) \le\max_{q\in W(B)}d_H(q,C)+\frac{\eps R_G}{4} \le\cost_{P_\eps}(C)+\frac{\eps R_G}{4}.
\]
Taking a maximum over $p\in P$ and using $P_\eps\subseteq P$ yields
\begin{equation}
    \cost_{P_\eps}(C) \le\cost_P(C) \le\cost_{P_\eps}(C)+\frac{\eps R_G}{4}.
    \label{eq:app-large-uniform-additive}
\end{equation}
No sum of shell errors is incurred.

Let $C_{P_\eps}$ be any exact optimal continuous $k$-center solution for $P_\eps$. Because $P_\eps$ is nonempty, this center set is nonempty. Since $P_\eps\subseteq P$, we have $\OPT_k(P_\eps)\le\OPT_k(P)$. Applying~\eqref{eq:app-large-uniform-additive} and $R_G\le4\OPT_k(P)$ from Lemma~\ref{lem:gonzalez} gives
\[
    \cost_P(C_{P_\eps}) \le\cost_{P_\eps}(C_{P_\eps})+\frac{\eps R_G}{4} =\OPT_k(P_\eps)+\frac{\eps R_G}{4} \le\OPT_k(P)+\eps\OPT_k(P) =(1+\eps)\OPT_k(P).
\]
Since $C_{P_\eps}$ was arbitrary, the guarantee holds for every exact optimum of the reduced instance.

\subparagraph{Size bound.}
There are at most $kL$ nonempty shells, where $L=\lceil8/\eps\rceil$. Using $R_G\ge R_0$, $0<\eps<1$, and $e^x-1\ge x$ for $x\ge0$,
\[
    \gamma_G =\min\{1,e^{\eps R_G/16}-1\} \ge\min\left\{1,\frac{\eps R_G}{16}\right\} \ge\eps\min\left\{1,\frac{R_0}{16}\right\} =c_0\eps.
\]
Consequently,
\[
    |W(B)|\le T_{k,D}(\gamma_G) \le k!\left(\frac{A}{c_0\eps}\right)^{kD}.
\]
Adding the anchors gives
\[
    |P_\eps| \le k+kL\,T_{k,D}(\gamma_G) \le k+k\left\lceil\frac{8}{\eps}\right\rceil k!\left(\frac{A}{c_0\eps}\right)^{kD} =\left(\frac{k}{\eps}\right)^{O(kD)}.
\]
In particular, the size bound is independent of $n$ and $R_G$.

\subparagraph{Construction time.}
The farthest-first anchors and their nearest-anchor partition take $O(nkD)$ time when the cost of a distance evaluation is
written explicitly. For fixed $D$, transforming the points and computing their radial coordinates take $O(n)$ time. With the intervals used in Section~\ref{sec:radial-shells}, the shell index of a non-anchor point $p\in P_i$ is
\[
    j(p)= \min\left\{L,\,1+ \left\lfloor\frac{r_i(p)}{\Delta}\right\rfloor \right\}.
\]
Initializing one list per anchor and radial interval takes $O(kL)$ time, and assigning all points to these lists takes $O(n)$ additional time.

Each non-anchor point belongs to exactly one shell, so $\sum_B|B|=|P\setminus A_k|\le n$. Summing the certificate construction and duplicate-removal costs over the nonempty shells gives
\[
    \sum_B O\bigl(|B|\,kD\,T_{k,D}(\gamma_G)\bigr) =O\bigl(nkD\,T_{k,D}(\gamma_G)\bigr).
\]
Collecting the original input references and the anchors takes $O(n+k)$ time and does not require geometric duplicate removal between different shells. The total construction time is therefore
\[
    O\bigl(nkD+kL+nkD\,T_{k,D}(\gamma_G)\bigr) = O\left(nkD\left(\frac{k}{\eps}\right)^{O(kD)}\right).
\]
Here, $L=O(1/\eps)$ and $T_{k,D}(\gamma_G)\le k!(A/(c_0\eps))^{kD}$, so the displayed bound absorbs the shell-list initialization cost as well. This proves all claims of Theorem~\ref{thm:large}.
\end{proof}


\section{Proof of the Coreset-Size Lower Bound}\label{app:coreset-lower-bound}
We prove Theorem~\ref{thm:angular-lower-bound} by constructing an instance with $\OPT_k(P)=1$ for which every nonempty weak $\eps$-coreset equals $P$. The auxiliary anchor points introduced below are input points of the lower-bound construction and are unrelated to the farthest-first anchors used in the coreset algorithm.

We begin by choosing sufficiently many separated directions for a single witness family. For $u,v\in\mathbb S^{D-1}$, write $\angle(u,v):=\arccos\langle u,v\rangle$.

\begin{lemma}\label{lem:separated-directions}
Let $D\ge2$ and $0<\alpha\le1$. There is a finite set $U\subseteq\mathbb S^{D-1}$ such that $\angle(u,v)\ge\alpha$ for all distinct $u,v\in U$ and 
\[
    |U|\ge c_D^{\mathrm{pack}}\alpha^{-(D-1)},
\]
where $c_D^{\mathrm{pack}}>0$ depends only on $D$.
\end{lemma}

\begin{proof}
Choose an inclusion-maximal $\alpha$-separated set $U\subseteq\mathbb S^{D-1}$. Such a set is finite by compactness of the sphere. Maximality implies that the spherical caps of angular radius $\alpha$ centered at the points of $U$ cover $\mathbb S^{D-1}$.

Let $\omega_j$ be the surface measure of the unit sphere $\mathbb S^j$, with $\omega_0=2$. A cap of angular radius $\alpha$ has measure
\[
    \omega_{D-2}\int_0^\alpha\sin^{D-2}t dt \le\omega_{D-2}\int_0^\alpha t^{D-2} dt =\frac{\omega_{D-2}}{D-1}\alpha^{D-1}.
\]
Since the caps centered at $U$ cover the sphere,
\[
    \omega_{D-1} \le |U|\frac{\omega_{D-2}}{D-1}\alpha^{D-1}.
\]
Consequently,
\[
    |U|\ge \frac{(D-1)\omega_{D-1}}{\omega_{D-2}} \alpha^{-(D-1)}.
\]
Thus one may take
\[
    c_D^{\mathrm{pack}} := \frac{(D-1)\omega_{D-1}}{\omega_{D-2}}.
\]
\end{proof}

For $D=1$, use the two directions $\mathbb S^0=\{-1,+1\}$ directly. This gives a constant number of directions, consistent with $(D-1)/2=0$.

We now use these directions to construct a family in which each input point is the unique farthest point for a suitable query center. Fix the constants
\[
    \Lambda:=\frac{2\sinh(3/2)}{\sinh^2(1/2)}, \qquad \eps_0:=\min\left\{\frac14,\frac{1}{4\Lambda}\right\}.
\]
Let $0<\eps\le\eps_0$ and put
\[
    \alpha:=2\sqrt{\Lambda\eps}\le1.
\]
For $D\ge2$, choose $U$ as in Lemma~\ref{lem:separated-directions} with this value of $\alpha$. For $D=1$, put $U:=\mathbb S^0$. In hyperbolic polar coordinates around a base point $o$, define
\[
    x_u:=\left(\frac12,u\right), \qquad c_u:=\left(\frac12+2\eps,-u\right) \qquad(u\in U),
\]
and let $X:=\{x_u:u\in U\}$. The points $x_u$ are input points, whereas the points $c_u$ are query centers and are not added to the input.

\begin{lemma}\label{lem:unique-maximum-witness}
For every $u\in U$,
\[
    d_H(x_u,c_u)=1+2\eps, \qquad d_H(x_v,c_u)\le1 \quad\text{for every }v\in U\setminus\{u\}.
\]
In particular, $x_u$ is the unique farthest point of $X$ from $c_u$.
\end{lemma}

\begin{proof}
Put $r:=1/2$ and $s:=1/2+2\eps$. The points $x_u$ and $c_u$ lie on opposite rays of one geodesic through $o$, so
\[
    d_H(x_u,c_u)=r+s=1+2\eps.
\]

Suppose first that $D\ge2$. For $v\ne u$, put $\theta:=\angle(u,v)\ge\alpha$. The angle between $v$ and $-u$ is $\pi-\theta$. By Lemma~\ref{lem:polar-cosine-law},
\[
\begin{aligned}
    \cosh d_H(x_v,c_u)
    &=\cosh r\cosh s-\sinh r\sinh s\cos(\pi-\theta)\\
    &=\cosh(r+s)-\sinh r\sinh s(1-\cos\theta)\\
    &=\cosh(1+2\eps)
      -\sinh(1/2)\sinh(1/2+2\eps)(1-\cos\theta).
\end{aligned}
\]
Since $0<\alpha\le1$ and $\alpha\le\theta\le\pi$,
\[
    1-\cos\theta \ge1-\cos\alpha \ge\frac{\alpha^2}{4} =\Lambda\eps.
\]
Also,
\[
    \sinh(1/2)\sinh(1/2+2\eps)\ge\sinh^2(1/2).
\]
The mean-value theorem and $\eps\le1/4$ give
\[
    \cosh(1+2\eps)-\cosh 1 \le2\eps\sinh(3/2) =\Lambda\eps\sinh^2(1/2).
\]
Combining these inequalities yields
\[
    \cosh d_H(x_v,c_u) \le\cosh(1+2\eps)-\Lambda\eps\sinh^2(1/2) \le\cosh 1.
\]
Since $\cosh$ is increasing on $[0,\infty)$, we obtain $d_H(x_v,c_u)\le1$.

For $D=1$, the only other direction is $v=-u$. Then $x_v$ and $c_u$ lie on the same ray, so
\[
    d_H(x_v,c_u)=s-r=2\eps\le1.
\]
\end{proof}

For $D\ge2$, the packing bound gives
\[
    |X|=|U| \ge c_D^{\mathrm{pack}}(2\sqrt{\Lambda})^{-(D-1)} \eps^{-(D-1)/2}.
\]
Thus, for every $D\ge1$, there is a constant $c_D'>0$ depending only on $D$ such that
\begin{equation}
    |X|=|U|\ge c_D'\eps^{-(D-1)/2}, 
    \label{eq:angular-family-size}
\end{equation}
where one may take $c_1'=2$ for $D=1$. The base point satisfies
\[
    \cost_X(\{o\})=\max_{x\in X}d_H(x,o)=\frac12.
\]
Hence the full family can be covered by one center at radius $1/2$, while each $c_u$ covers $X\setminus\{x_u\}$ at radius $1$ and leaves $x_u$ at distance $1+2\eps$.

To turn these single-center queries into obstructions for weak $k$-center coresets, we combine distant copies of $X$ with an auxiliary anchor pair. Fix $k\ge2$ and choose base points $o_0,o_1,\ldots,o_{k-1}$ on a common hyperbolic geodesic with
\[
    d_H(o_i,o_j)\ge20\qquad(i\ne j).
\]
Place two auxiliary anchor points $b^-,b^+$ on this geodesic with midpoint $o_0$ and
\[
    d_H(b^-,o_0)=d_H(b^+,o_0)=1.
\]
Hence $d_H(b^-,b^+)=2$.

For each $1\le i\le k-1$, take an isometric copy $X_i=\{x_{i,u}:u\in U\}$ of $X$ centered at $o_i$. Apply the same isometry to the query centers, and denote their images by $c_{i,u}$.
Then
\[
    d_H(o_i,x_{i,u})=\frac12, \qquad d_H(o_i,c_{i,u})=\frac12+2\eps\le1,
\]
and Lemma~\ref{lem:unique-maximum-witness} gives
\[
    d_H(x_{i,u},c_{i,u})=1+2\eps, \qquad d_H(x_{i,v},c_{i,u})\le1\quad(v\ne u).
\]
Define
\[
    P:=\{b^-,b^+\}\cup\bigcup_{i=1}^{k-1}X_i.
\]
The input groups are pairwise disjoint. The triangle inequality gives the separation estimates
\begin{equation}
\begin{aligned}
    d_H(x,y)&\ge19
        &&(x\in X_i,\ y\in X_j,\ i\ne j),\\
    d_H(x,b^\pm)&\ge\frac{37}{2}
        &&(x\in X_i).
\end{aligned}
\label{eq:angular-group-separation}
\end{equation}
These estimates first determine the optimum radius of the full input.

\begin{lemma}\label{lem:lower-bound-optimum}
The constructed instance satisfies $\OPT_k(P)=1$.
\end{lemma}

\begin{proof}
The center set $C^\star:=\{o_0,o_1,\ldots,o_{k-1}\}$ covers the anchor pair at radius $1$ and each $X_i$ at radius $1/2$. Thus $\cost_P(C^\star)\le1$, proving $\OPT_k(P)\le1$.

For the reverse inequality, choose $y_i\in X_i$ for every $1\le i\le k-1$. The set
\[
    Y:=\{b^-,b^+,y_1,\ldots,y_{k-1}\}
\]
contains $k+1$ points at pairwise distances at least $2$. Indeed, the two anchors have distance $2$, and all other pairs have larger distances by~\eqref{eq:angular-group-separation}. By the triangle inequality, a ball of radius strictly less than $1$ contains at most one point of $Y$, regardless of its center. Therefore at most $k$ such balls cannot cover $Y$, and $\OPT_k(P)\ge1$.
\end{proof}

We next show that every weak coreset must retain both auxiliary anchors and at least one point from each witness family.
The same packing argument will then force its optimum radius to remain $1$, allowing us to use the query centers above as part of exact optima for the reduced instance.

\begin{lemma}\label{lem:mandatory-regions}
Let $Q\subseteq P$ be a nonempty weak $\eps$-coreset. Then $b^-,b^+\in Q$ and $Q\cap X_i\ne\varnothing$ for every $1\le i\le k-1$. Moreover, $\OPT_k(Q)=1$.
\end{lemma}

\begin{proof}
Write $r_Q:=\OPT_k(Q)\le\OPT_k(P)=1$. We use the following observation about exact optima of $Q$. Assign each point of $Q$ to one of its nearest centers in an exact optimum, breaking ties arbitrarily, and discard every center assigned no point. The resulting center set remains nonempty and uses at most $k$ centers. It still covers $Q$ at radius $r_Q$, so it remains an exact optimum. Every remaining center is within distance $r_Q$ of an assigned point of $Q$.

Suppose first that $b^-\notin Q$. Place a center at $b^+$ if $b^+\in Q$, and a center at $o_i$ for every nonempty $Q\cap X_i$. This uses at most $1+(k-1)=k$ centers and covers $Q$ at radius at most $1/2$. Hence $r_Q\le1/2$.

Choose an exact optimum $\widehat C$ after discarding unassigned centers as above. For each $\widehat c\in\widehat C$, choose an assigned point $q\in Q$ with $d_H(q,\widehat c)\le r_Q\le1/2$. If $q=b^+$, then
\[
    d_H(b^-,\widehat c) \ge d_H(b^-,b^+)-d_H(b^+,\widehat c) \ge2-\frac12=\frac32.
\]
If $q\in X_i$, then~\eqref{eq:angular-group-separation} gives
\[
    d_H(b^-,\widehat c) \ge d_H(b^-,q)-d_H(q,\widehat c) \ge\frac{37}{2}-\frac12=18.
\]
Consequently,
\[
    \cost_P(\widehat C) \ge d_H(b^-,\widehat C) \ge\frac32 >1+\eps =(1+\eps)\OPT_k(P),
\]
where $\eps\le1/4$. This exact optimum of $Q$ violates the weak coreset guarantee. Thus $b^-\in Q$, and the same argument gives $b^+\in Q$.

Now suppose that $Q\cap X_i=\varnothing$ for some $i$. Use the two anchor points themselves as centers and place a center at $o_j$ for every nonempty $Q\cap X_j$, with $j\ne i$. This uses at most $2+(k-2)=k$ centers and covers $Q$ at radius at most $1/2$. Thus again $r_Q\le1/2$.

Choose an exact optimum $\widehat C$ with all unassigned centers removed, and fix any $x\in X_i$. Every remaining center $\widehat c$ has an assigned point $q\in Q$ from a different input group. By~\eqref{eq:angular-group-separation},
\[
    d_H(x,\widehat c) \ge d_H(x,q)-d_H(q,\widehat c) \ge\frac{37}{2}-\frac12=18.
\]
Hence
\[
    \cost_P(\widehat C) \ge d_H(x,\widehat C) \ge18>(1+\eps)\OPT_k(P),
\]
again contradicting the weak coreset guarantee. Therefore $Q\cap X_i\ne\varnothing$ for every $i$.

Finally, choose $q_i\in Q\cap X_i$ for $1\le i\le k-1$. The $k+1$ retained points $\{b^-,b^+,q_1,\ldots,q_{k-1}\}$ have pairwise distances at least $2$. The packing argument from Lemma~\ref{lem:lower-bound-optimum} yields $r_Q\ge1$. Since $Q\subseteq P$ also gives $r_Q\le1$, we conclude that $\OPT_k(Q)=1$.
\end{proof}

With the reduced optimum fixed at $1$, we can now show that retaining one point per witness family is not sufficient. Omitting any individual witness produces an exact optimum of the reduced instance that violates the required guarantee on $P$.

\begin{lemma}\label{lem:every-witness-necessary}
Let $Q\subseteq P$ be a nonempty weak $\eps$-coreset. Then $X_i\subseteq Q$ for every $1\le i\le k-1$.
\end{lemma}

\begin{proof}
Suppose that $x_{i,u}\notin Q$ for some $i$ and $u\in U$. Lemma~\ref{lem:mandatory-regions} gives $\OPT_k(Q)=1$. Consider
\[
    C_{i,u}:= \{o_0,c_{i,u}\} \cup\{o_j:1\le j\le k-1,\ j\ne i\}.
\]
These are exactly $k$ distinct centers. The center $o_0$ covers the anchors at radius $1$. For $j\ne i$, the center $o_j$ covers $X_j$ at radius $1/2$. The center $c_{i,u}$ covers $X_i\setminus\{x_{i,u}\}$ at radius at most $1$ by Lemma~\ref{lem:unique-maximum-witness}. Since $x_{i,u}\notin Q$,
\[
    \cost_Q(C_{i,u})\le1=\OPT_k(Q).
\]
Therefore $\cost_Q(C_{i,u})=\OPT_k(Q)=1$, and $C_{i,u}$ is an exact optimal continuous $k$-center solution for $Q$.

On the omitted point,
\[
    d_H(x_{i,u},c_{i,u})=1+2\eps.
\]
Every other center in $C_{i,u}$ is a base point $o_j$ with $j\in\{0,\ldots,k-1\}\setminus\{i\}$, so
\[
    d_H(x_{i,u},o_j) \ge d_H(o_i,o_j)-d_H(o_i,x_{i,u}) \ge20-\frac12>1+2\eps.
\]
Thus
\[
    d_H(x_{i,u},C_{i,u})=1+2\eps.
\]
Using Lemma~\ref{lem:lower-bound-optimum}, we obtain
\[
    \cost_P(C_{i,u}) \ge1+2\eps >(1+\eps)\OPT_k(P).
\]
We have exhibited an exact optimum of $Q$ that fails the required approximation on $P$, contradicting the weak coreset guarantee. Thus every point of every witness family belongs to $Q$.
\end{proof}

The preceding lemmas show that every input point is mandatory. It remains to count the points in the construction.

\begin{proof}[Proof of Theorem~\ref{thm:angular-lower-bound}]
Let $Q\subseteq P$ be any nonempty weak $\eps$-coreset. Lemmas~\ref{lem:mandatory-regions} and~\ref{lem:every-witness-necessary} show that $Q=P$. By~\eqref{eq:angular-family-size} and $k\ge2$,
\[
    |Q|=|P| = 2+(k-1)|U| \ge(k-1)c_D'\eps^{-(D-1)/2} \ge\frac{c_D'}{2}\,k\eps^{-(D-1)/2}.
\]
Taking $c_D:=c_D'/2$ proves the theorem.
\end{proof}

The constructed input also satisfies the separation promise of Section~\ref{sec:sep-inputs}. Its exactly $k$ blocks are the anchor pair and the $k-1$ witness families. Their maximum enclosing radius is $1$, and distinct blocks have distance at least $37/2>2$ by~\eqref{eq:angular-group-separation}. Thus the same instance provides the lower bound used in Corollary~\ref{cor:sep-tight}. We conclude by explaining why the restriction $k\ge2$ is needed for a lower bound that diverges as $\eps$ decreases.

\begin{remark}[The case $k=1$]
When $k=1$, every finite nonempty set $S\subseteq\Hd$ admits an exact weak coreset of size at most $D+1$. Indeed, applying Helly's theorem to the balls $B_H(p,r)$ for radii approaching $\operatorname{rad}(S)$ from below yields a nonempty subset $Q\subseteq S$ with $|Q|\le D+1$ and $\operatorname{rad}(Q)=\operatorname{rad}(S)$ \cite[Corollary~4.3]{ledyaev2006helly}. The circumcenter of $S$ is therefore also optimal for $Q$. Since circumcenters in $\Hd$ are unique \cite[Proposition~II.2.7]{bridson2013metric}, the two optimal centers coincide. Thus, for fixed $D$, no weak coreset-size lower bound can diverge as $\eps\to0$ when $k=1$. This does not imply preservation of farthest-point distances at arbitrary query centers.
\end{remark}


\section{Proof of the Local-Dimensional Lower Bound}\label{app:local-dimensional-lower}
We prove Theorem~\ref{thm:local-dimensional-lower}. Fix integers $D\ge1$ and $k\ge4D$, and let $0<\eps\le1/128D$. We first construct a Euclidean grid for which each input point can be exposed by a set of $2D$ query centers, and then transfer this configuration to a small hyperbolic region. Distant copies and auxiliary anchor pairs will make the exposing queries exact optima of the reduced instance. Throughout the proof, all centers remain unrestricted in $\Hd$. The query centers are not included in the input, and the auxiliary anchor pairs are unrelated to the farthest-first anchors of the coreset algorithm.

Put $L:=1+2\eps$, and let $G$ be the Cartesian grid whose coordinates belong to
\[
    \left\{-\frac{1}{16D}+8\eps j: j=0,1,\ldots, \left\lfloor\frac{1}{64D\eps}\right\rfloor\right\}.
\]
Then $G\subseteq[-1/(16D),1/(16D)]^D$ and
\begin{equation}
    |G|= \left(1+\left\lfloor\frac{1}{64D\eps}\right\rfloor\right)^D \ge(64D\eps)^{-D}.
    \label{eq:local-grid-cardinality}
\end{equation}
For each $p\in G$, define the $2D$ Euclidean query centers
\[
    V_p:=\{p+Le_j,p-Le_j:1\le j\le D\},
\]
where $e_1,\ldots,e_D$ are the standard coordinate vectors. Write $d_E(x,V):=\min_{v\in V}\|x-v\|_2$ for the Euclidean distance from a point to a nonempty finite set. Every point of $V_p$ has distance $L$ from $p$, so $d_E(p,V_p)=L$.

For $q\in G\setminus\{p\}$, put $t:=\|q-p\|_\infty$. The grid spacing and the enclosing cube give
\[
    8\eps\le t\le\frac{1}{8D}, \qquad \|q-p\|_2^2\le Dt^2\le\frac{t}{8}.
\]
Minimizing over the signed coordinate directions yields
\[
    d_E(q,V_p)^2=L^2+\|q-p\|_2^2-2Lt.
\]
Since $L\ge1$ and $L^2=(1+2\eps)^2\le1+5\eps$ for $\eps\le1/4$, we obtain
\begin{equation}
    d_E(q,V_p)^2 \le 1+5\eps+\frac{t}{8}-2t =1+5\eps-\frac{15}{8}t \le 1-10\eps.
\label{eq:euclidean-grid-gap}
\end{equation}
Thus $V_p$ covers every other grid point at a distance strictly less than $1$, while its distance from $p$ is $1+2\eps$.

We next realize this configuration in the Poincar\'e ball. Put
\[
    s:=\frac{\sqrt{\eps}}{4}, \qquad R:=2s,
\]
and define $f(x):=sx$ for $x\in B_2(0,2)$. Since $2s<1$, the image of this Euclidean ball lies in the Poincar\'e ball. Every grid point and query center belongs to $B_2(0,2)$, since
\[
    \|p\|_2\le\frac{1}{16\sqrt D}, \qquad \|v\|_2\le1+2\eps+\frac{1}{16\sqrt D}<2 \quad(v\in V_p).
\]
For any $x,y\in B_2(0,2)$, the Poincar\'e length element $2\|dz\|_2/(1-\|z\|_2^2)$ gives
\begin{equation}
    R\|x-y\|_2 \le d_H(f(x),f(y)) \le\frac{R}{1-4s^2}\|x-y\|_2 \le R(1+\eps)\|x-y\|_2.
    \label{eq:local-poincare-distortion}
\end{equation}
For the lower bound, every path from $f(x)$ to $f(y)$ has hyperbolic length at least twice its Euclidean length. For the upper bound, use the Euclidean segment between these points, which stays inside $B_2(0,2s)$. The last inequality follows from $4s^2=\eps/4$ and $(1-\eps/4)^{-1}\le1+\eps$. The map $f$ constructs a new local configuration by contracting Euclidean coordinates; it is not a hyperbolic isometry or a rescaling of the hyperbolic metric.

Set $X:=f(G)$ and, for each $p\in G$, define the hyperbolic query-center set
\[
    \mathcal C_{f(p)}:=f(V_p).
\]
Equations~\eqref{eq:euclidean-grid-gap} and~\eqref{eq:local-poincare-distortion} imply 
\begin{align*}
    d_H(f(q),\mathcal C_{f(p)})
    &\le R(1+\eps)\sqrt{1-10\eps}<R
       &&(q\in G\setminus\{p\}),\\
    d_H(f(p),\mathcal C_{f(p)})
    &\ge(1+2\eps)R.
\end{align*}
Here $1-10\eps>0$, and the strict upper bound follows from
\[
    (1+\eps)^2(1-10\eps) =1-8\eps-19\eps^2-10\eps^3<1.
\]
Applying~\eqref{eq:local-poincare-distortion} with one point at the origin also gives
\[
    d_H(0,f(p)) \le\frac{R(1+\eps)}{16\sqrt D}<R \qquad(p\in G).
\]
Moreover, our bound on $\eps$ implies
\[
    (1+\eps)\left(1+2\eps+\frac{1}{16\sqrt D}\right)<2,
\]
so every query center lies within distance $2R$ of the origin. Consequently,
\begin{equation}
\begin{gathered}
    |X|=|G|,\qquad X\subseteq B_H(0,R),\\
    |\mathcal C_x|=2D,\qquad \mathcal C_x\subseteq B_H(0,2R),\\
    \cost_{X\setminus\{x\}}(\mathcal C_x)<R, \qquad d_H(x,\mathcal C_x)\ge(1+2\eps)R \qquad(x\in X).
\end{gathered}
\label{eq:local-component-properties}
\end{equation}

We now combine distant copies of this component with auxiliary anchor pairs to enforce exact optimality. Put
\[
    m:=2D, \qquad g:=k-2m+1=k-4D+1\ge1.
\]
Choose $g+m$ base points on a common hyperbolic geodesic with pairwise distances at least $40R$. At $g$ of these points, place isometric copies $X_1,\ldots,X_g$ of $X$ with base points $o_1,\ldots,o_g$. Apply the same isometries to the associated query centers. For each $p\in X_i$, denote the resulting query-center set by $\mathcal C_{i,p}$. The properties in~\eqref{eq:local-component-properties} are preserved by these isometries. At each remaining base point $z_j$, place an anchor pair $A_j=\{a_j,b_j\}$ with midpoint $z_j$ and
\[
    d_H(a_j,z_j)=d_H(b_j,z_j)=R, \qquad d_H(a_j,b_j)=2R.
\]
Define the input by
\[
    P:=\bigcup_{i=1}^g X_i\;\cup\;\bigcup_{j=1}^m A_j.
\]
Each input group $X_i$ or $A_j$ lies in the radius-$R$ ball about its base point. Thus points in different groups have distance at least $40R-R-R=38R$, and all groups are disjoint.

The $g$ component base points and $m$ anchor midpoints cover $P$ at radius $R$ using
\[
    g+m=k-m+1\le k
\]
centers. Hence $\OPT_k(P)\le R$. Conversely, the $2m$ anchor endpoints together with one point from each component form
\[
    2m+g=k+1
\]
points at pairwise distances at least $2R$. By the triangle inequality, any ball of radius strictly less than $R$ contains at most one of these points, regardless of its center. Thus $k$ such balls cannot cover $P$, proving
\[
    \OPT_k(P)=R.
\]

Let $Q\subseteq P$ be a nonempty weak $\eps$-coreset, and put $r_Q:=\OPT_k(Q)\le\OPT_k(P)=R$. We first show that $Q$ must contain both endpoints of every anchor pair. Suppose that $b_j\notin Q$, and take an exact optimum $C$ for $Q$. Remove all centers in the closed ball $B_H(b_j,3R)$, and insert $a_j$ if $a_j\in Q$. Denote the resulting center set by $C'$. If $a_j\in Q$, a center covering it at radius $r_Q$ has distance at most $2R+r_Q\le3R$ from $b_j$ and is therefore removed. Consequently, inserting $a_j$ does not increase the number of centers, and $|C'|\le|C|\le k$.

No removed center can cover a point of another input group at radius $r_Q$. Indeed, for such a point $q$ and a removed center $c$,
\[
    d_H(q,c) \ge d_H(q,b_j)-d_H(b_j,c) \ge38R-3R=35R>r_Q.
\]
The only possible retained point of $A_j$ is $a_j$, which is covered by the inserted center when necessary. Thus $C'$ still covers $Q$ at radius $r_Q$. It is nonempty because $Q$ is nonempty, and it remains an exact optimum of $Q$. Every surviving center of $C$ is farther than $3R$ from $b_j$, and the only possible new center $a_j$ has distance $2R$ from $b_j$. Therefore
\[
    \cost_P(C') \ge d_H(b_j,C') \ge 2R >(1+\eps)R =(1+\eps)\OPT_k(P),
\]
contradicting the weak coreset guarantee. The same argument applies to an omitted $a_j$, so all anchor endpoints belong to $Q$.

Next, suppose that $Q\cap X_i=\varnothing$ for some component. Take an exact optimum for $Q$, assign each point of $Q$ to a nearest center, and remove every center assigned no point. Call the remaining nonempty center set $\widehat C$. It is still an exact optimum of $Q$. For each $\widehat c\in\widehat C$, choose an assigned point $q\in Q$ with $d_H(q,\widehat c)\le r_Q\le R$. Since $q$ belongs to a different input group from any $x\in X_i$,
\[
    d_H(x,\widehat c) \ge d_H(x,q)-d_H(q,\widehat c) \ge38R-R=37R.
\]
Consequently,
\[
    \cost_P(\widehat C) \ge37R>(1+\eps)\OPT_k(P),
\]
again contradicting the weak coreset guarantee. Thus $Q$ meets every component.

We have shown that $Q$ contains all $2m$ anchor endpoints and at least one point from each of the $g$ components. These $2m+g=k+1$ points are pairwise at distance at least $2R$. The preceding packing argument therefore gives $r_Q\ge R$. Together with $r_Q\le R$, this proves
\begin{equation}
    \OPT_k(Q)=R.
    \label{eq:local-reduced-optimum}
\end{equation}

With the reduced optimum fixed at $R$, we can force every component point to be retained. Suppose that $p\in X_i\setminus Q$, and define
\[
    C_p:= \mathcal C_{i,p} \cup\{o_\ell:1\le\ell\le g,\ \ell\ne i\} \cup\{z_j:1\le j\le m\}.
\]
These are exactly
\[
    m+(g-1)+m=2m+g-1=k
\]
distinct centers. The centers in $\mathcal C_{i,p}$ lie within distance $2R$ of $o_i$, while every other center is a different base point at distance at least $40R$ from $o_i$. By~\eqref{eq:local-component-properties}, the set $\mathcal C_{i,p}$ covers $X_i\setminus\{p\}$ at radius less than $R$. The other component base points cover their components at radius $R$, and each anchor midpoint covers its pair at radius $R$. Since $p\notin Q$, we obtain
\[
    \cost_Q(C_p)\le R=\OPT_k(Q).
\]
Thus $C_p$ is an exact optimum of $Q$.

On the omitted point, the local construction gives $d_H(p,\mathcal C_{i,p})\ge(1+2\eps)R$. Every remaining center of $C_p$ is a base point of a different group and is therefore at distance at least $40R-R=39R$ from $p$. Hence
\[
    \cost_P(C_p) \ge d_H(p,C_p) \ge(1+2\eps)R >(1+\eps)\OPT_k(P).
\]
This contradicts the weak coreset guarantee. Therefore every point of every component belongs to $Q$. Since all anchor endpoints also belong to $Q$, we conclude that $Q=P$.

Finally, the disjointness of the input groups and \eqref{eq:local-grid-cardinality} give
\[
    |Q|=|P| =g|G|+2m \ge(64D)^{-D}(k-4D+1)\eps^{-D}.
\]
For $k\ge8D$, we have $k-4D+1\ge k/2$, so the required coreset size is $\Omega_D(k\eps^{-D})$. This proves Theorem~\ref{thm:local-dimensional-lower}.


\section{Proof of the Joint-Parameter Lower Bound}\label{app:joint-parameter-lower}
We begin with the tree construction of Kluk et al.~\cite[Theorem~1.12]{KlukEtAl2026}. Take a full binary tree rooted at $r$, with unit-length edges and all leaves at depth $h$. For each nonroot vertex $u$ at depth $j$, attach a new path of $j$ unit-length edges at $u$ and denote its other endpoint by $\widehat u$. The added paths are disjoint from one another and from the original tree except at their attachment vertices. Let $T$ be the resulting tree and let $L$ be the set of the $2^h$ leaves of the original binary tree. For $v\in L$, write $\pi_v$ for its root-to-leaf path in the original tree, and let $Y_v$ contain the sibling of each nonroot vertex on $\pi_v$. Define
\[
    Z_v:=\{\widehat u:u\in Y_v\}.
\]
There is one such sibling at each depth from $1$ to $h$, so $|Z_v|=h$. The leaves in $L$ will give the input points of the component, whereas the vertices in $Z_v$ will give its query centers.

For every nonroot vertex $u$ of the original tree and every $v\in L$,
\begin{equation}
    d_T(v,\widehat u)=h+2d_T(u,\pi_v),
    \label{eq:tree-distance-identity}
\end{equation}
where $d_T(u,\pi_v)$ denotes the minimum tree distance from $u$ to a vertex of $\pi_v$. To verify the identity, let $u$ have depth $j$ and let $w$ be the lowest common ancestor of $u$ and $v$. The path from $v$ to $\widehat u$ goes through $w$ and $u$, so
\[
    d_T(v,\widehat u) =(h-\operatorname{depth}(w)) +(j-\operatorname{depth}(w))+j =h+2(j-\operatorname{depth}(w)).
\]
Since $j-\operatorname{depth}(w)=d_T(u,\pi_v)$, this proves~\eqref{eq:tree-distance-identity}. Every vertex of $Y_v$ has distance $1$ from $\pi_v$. Consequently,
\[
    d_T(v,Z_v)=h+2.
\]
For any other leaf $v'\in L\setminus\{v\}$, consider the child towards $v'$ at the first divergence of $\pi_v$ and $\pi_{v'}$. This child belongs to $Y_v$ and lies on $\pi_{v'}$. Applying~\eqref{eq:tree-distance-identity} to it gives
\[
    d_T(v',Z_v)\le h.
\]
Moreover,
\[
    d_T(r,v)=h\quad(v\in L), \qquad d_T(r,z)\le2h\quad(z\in Z_v).
\]
Thus the $h$ query vertices in $Z_v$ cover all other input leaves at radius $h$ while leaving $v$ at distance $h+2$.

We next transfer these distance gaps to hyperbolic space. Put $\eta:=1/(h+2)$. By Sarkar's embedding theorem~\cite[Theorem~6]{Sarkar2011}, there are an embedding $f:V(T)\to\mathbb H^2$ and a common scale $\lambda>0$ such that
\begin{equation}
    \frac{\lambda}{1+\eta}d_T(u,v) \le d_H(f(u),f(v)) \le\lambda d_T(u,v) \qquad(u,v\in V(T)).
    \label{eq:scaled-tree-embedding}
\end{equation}
The embedding includes the input leaves, the root, and all query vertices. The scale $\lambda$ is chosen as part of the embedding construction; no rescaling isometry of hyperbolic space is assumed.
Set
\[
    R:=\lambda h, \qquad o:=f(r), \qquad X:=f(L),
\]
and, for each $p=f(v)\in X$, define the query-center set
\[
    \mathcal C_p:=f(Z_v).
\]
The lower bound in~\eqref{eq:scaled-tree-embedding} makes $f$ injective, so $|X|=2^h$ and $|\mathcal C_p|=h$. The upper bound in that equation gives
\begin{equation}
    X\subseteq B_H(o,R), \qquad \mathcal C_p\subseteq B_H(o,2R), \qquad \cost_{X\setminus\{p\}}(\mathcal C_p)\le R \qquad(p\in X).
\label{eq:hyperbolic-component-properties}
\end{equation}
On the distinguished point, the lower bound gives
\begin{equation}
    d_H(p,\mathcal C_p) \ge\frac{\lambda(h+2)}{1+\eta} >(1+\eta)R \ge(1+\eps)R.
    \label{eq:hyperbolic-component-gap}
\end{equation}
Indeed,
\[
    (h+2)-h(1+\eta)^2 =\frac{3h+8}{(h+2)^2}>0,
\]
which proves the strict inequality because $R=\lambda h$. The last inequality uses $\eps\le\eta$. Thus the gap remains strict even when $\eps=1/(h+2)$.

To make these exposing queries exact optima of reduced instances, put
\[
    g:=k-2h+1\ge1.
\]
Choose $g+h$ base points on a common hyperbolic geodesic with pairwise distances at least $40R$. At $g$ of them, place isometric copies $X_1,\ldots,X_g$ of $X$ with base points $o_1,\ldots,o_g$. Apply the same isometries to the corresponding query centers, and for $p\in X_i$ denote the resulting set by $\mathcal C_{i,p}$. In particular, $X_i\subseteq B_H(o_i,R)$ and $\mathcal C_{i,p}\subseteq B_H(o_i,2R)$, and the covering and gap bounds in~\eqref{eq:hyperbolic-component-properties} and~\eqref{eq:hyperbolic-component-gap} are preserved. At each remaining base point $z_j$, place an auxiliary anchor pair $A_j=\{a_j,b_j\}$ with midpoint $z_j$ and
\[
    d_H(a_j,z_j)=d_H(b_j,z_j)=R, \qquad d_H(a_j,b_j)=2R.
\]
Define
\[
    P:=\bigcup_{i=1}^{g}X_i \;\cup\;\bigcup_{j=1}^{h}A_j.
\]
Each input group $X_i$ or $A_j$ lies in the radius-$R$ ball about its base point. Hence points in different groups have distance at least $40R-R-R=38R$, and all input groups are disjoint.

The component base points and anchor midpoints cover $P$ at radius $R$ using
\[
    g+h=k-h+1\le k
\]
centers. Therefore $\OPT_k(P)\le R$. For the reverse inequality, choose one point from each $X_i$. Together with all $2h$ anchor endpoints, these form
\[
    g+2h=k+1
\]
points at pairwise distances at least $2R$. By the triangle inequality, any ball of radius strictly less than $R$ contains at most one of them, regardless of its center. Thus $k$ such balls cannot cover $P$, and
\[
    \OPT_k(P)=R.
\]

Let $Q\subseteq P$ be a nonempty weak $\eps$-coreset, and put $r_Q:=\OPT_k(Q)\le\OPT_k(P)=R$. We first show that both endpoints of every anchor pair belong to $Q$. Suppose that $b_j\notin Q$, and take an exact optimum $C$ for $Q$. Remove all centers in the closed ball $B_H(b_j,3R)$, and insert $a_j$ if $a_j\in Q$. Denote the resulting center set by $C'$. If $a_j\in Q$, a center covering it at radius $r_Q$ has distance at most $2R+r_Q\le3R$ from $b_j$ and is therefore removed. Hence inserting $a_j$ does not increase the number of centers, and $|C'|\le|C|\le k$.

No removed center can cover a point of another input group at radius $r_Q$. Indeed, for any such point $q$ and any removed center $c$,
\[
    d_H(q,c) \ge d_H(q,b_j)-d_H(b_j,c) \ge38R-3R=35R>r_Q.
\]
The only possible retained point of $A_j$ is $a_j$, which is covered by the inserted center when necessary. Thus $C'$ still covers $Q$ at radius $r_Q$. It is nonempty because $Q$ is nonempty, and
\[
    \cost_Q(C')=r_Q=\OPT_k(Q).
\]
Every surviving center of $C$ is farther than $3R$ from $b_j$, and the only possible new center $a_j$ has distance $2R$ from $b_j$. Consequently,
\[
    \cost_P(C') \ge d_H(b_j,C') \ge 2R >(1+\eps)R = (1+\eps)\OPT_k(P),
\]
where $\eps\le1/(h+2)\le1/3$. This exact optimum of $Q$ violates the weak coreset guarantee. The same argument applies to an omitted $a_j$, so all anchor endpoints belong to $Q$.

Next, suppose that $Q\cap X_i=\varnothing$ for some component. Take an exact optimum for $Q$, assign each point of $Q$ to one of its nearest centers, breaking ties arbitrarily, and discard every center assigned no point. Call the remaining nonempty center set $\widehat C$. It still covers $Q$ at radius $r_Q$ and therefore remains an exact optimum. For each $\widehat c\in\widehat C$, choose an assigned point $q\in Q$ with $d_H(q,\widehat c)\le r_Q\le R$. For any $x\in X_i$, the points $x$ and $q$ belong to different input groups, so
\[
    d_H(x,\widehat c) \ge d_H(x,q)-d_H(q,\widehat c) \ge 38R-R=37R.
\]
Hence
\[
    \cost_P(\widehat C) \ge 37R>(1+\eps)\OPT_k(P),
\]
again contradicting the weak coreset guarantee. Thus $Q\cap X_i\ne\varnothing$ for every $i$.

The set $Q$ therefore contains all $2h$ anchor endpoints and at least one point from each of the $g$ components. These $2h+g=k+1$ points have pairwise distances at least $2R$. The same packing argument used for $P$ gives $r_Q\ge R$.
Together with $r_Q\le R$, this proves
\begin{equation}
    \OPT_k(Q)=R.
    \label{eq:reduced-optimum-forced}
\end{equation}

With the reduced optimum fixed at $R$, we can now force every component point to be retained. Suppose that $p\in X_i\setminus Q$, and consider
\[
    C_p:=\mathcal C_{i,p} \cup\{o_\ell:1\le\ell\le g,\ \ell\ne i\} \cup\{z_j:1\le j\le h\}.
\]
These are exactly
\[
    h+(g-1)+h=k
\]
distinct centers. Indeed, $\mathcal C_{i,p}$ contains $h$ distinct points within distance $2R$ of $o_i$, whereas every other center is a different base point at distance at least $40R$ from $o_i$. By~\eqref{eq:hyperbolic-component-properties}, $\mathcal C_{i,p}$ covers $X_i\setminus\{p\}$ at radius $R$. The other component base points cover their respective components at radius $R$, and the anchor midpoints cover all anchor pairs at radius $R$. Since $p\notin Q$,
\[
    \cost_Q(C_p)\le R=\OPT_k(Q).
\]
Thus $C_p$ is an exact optimum of $Q$.

On the omitted point,~\eqref{eq:hyperbolic-component-gap} gives
\[
    d_H(p,\mathcal C_{i,p})>(1+\eps)R.
\]
Every other center of $C_p$ is a base point of a different group and is therefore at distance at least $40R-R=39R$ from $p$. Consequently,
\[
    \cost_P(C_p) \ge d_H(p,C_p) >(1+\eps)R = (1+\eps)\OPT_k(P),
\]
contradicting the weak coreset guarantee. Therefore every component point belongs to $Q$. Since all anchor endpoints also belong to $Q$, we conclude that $Q=P$.

Finally, the input groups are disjoint, each component contains $2^h$ points, and there are $h$ anchor pairs. Hence
\[
    |Q|=|P|=g2^h+2h=(k-2h+1)2^h+2h.
\]
This proves Theorem~\ref{thm:joint-parameter-lower}.


\section{Proof of the Euclidean Additive Lower Bound}\label{app:sep-euclidean}
For $1\le i\le N$, put
\[
    u_i:=\bigl(\cos(2\pi(i-1)/N),\sin(2\pi(i-1)/N)\bigr).
\]
Define
\[
    \delta_N:=1-\cos(2\pi/N), \qquad R\ge\frac{16\alpha}{\delta_N},
\]
and set
\[
    p_i:=\frac R2u_i, \qquad S:=\{p_1,\ldots,p_N\},
\]
\[
    a:=9Re_1, \qquad b:=11Re_1, \qquad P:=S\cup\{a,b\},
\]
where $e_1=(1,0)$. The polygon vertices and the two anchors are distinct, so $|P|=N+2$. Since $0<\delta_N\le2$, the choice of $R$ implies $R\ge8\alpha$.

We first verify the separation condition and the optimum radius. The origin covers $S$ at radius $R/2$. Conversely, $\sum_{i=1}^N p_i=0$, so every $z\in\mathbb R^2$ satisfies
\[
    \max_{1\le i\le N}\|p_i-z\|_2^2 \ge\frac1N\sum_{i=1}^N\|p_i-z\|_2^2 =\frac{R^2}{4}+\|z\|_2^2 \ge\frac{R^2}{4}.
\]
Hence $\operatorname{rad}(S)=R/2$. The anchor pair has distance $\|a-b\|_2=2R$ and enclosing radius $R$, with midpoint $m:=10Re_1$. Every polygon vertex is at distance at least $9R-R/2=17R/2$ from either anchor. Thus the partition into exactly two nonempty blocks $S$ and $\{a,b\}$ satisfies~\eqref{eq:sep-condition} with maximum block radius $R$.

The centers $0$ and $m$ cover $P$ at radius $R$, proving $\OPT_2(P)\le R$. For any $p\in S$, the three points $a,b,p$ have pairwise distances at least $2R$. By the triangle inequality, a ball of radius strictly less than $R$ contains at most one of these points, regardless of its center. Therefore two such balls cannot cover $P$, and
\[
    \OPT_2(P)=R.
\]

We next associate a query-center set with each polygon vertex. For $1\le i\le N$, put
\[
    c_i:=-\left(\frac R2+2\alpha\right)u_i, \qquad C_i:=\{c_i,m\}.
\]
These are query centers and are not added to the input. The point $p_i$ lies on the ray opposite to $c_i$, so
\[
    \|p_i-c_i\|_2=R+2\alpha.
\]
For $j\ne i$, the equally spaced directions satisfy $1-\langle u_i,u_j\rangle\ge\delta_N$. Consequently,
\begin{align*}
    \|p_j-c_i\|_2^2
    &=(R+2\alpha)^2 -\left(\frac{R^2}{2}+2\alpha R\right) \bigl(1-\langle u_i,u_j\rangle\bigr)\\
    &\le R^2+4\alpha R+4\alpha^2 -\frac{R^2}{2}\delta_N\\
    &\le R^2-4\alpha R+4\alpha^2\\
    &\le R^2.
\end{align*}
The penultimate inequality uses $R\delta_N\ge16\alpha$, and the last uses $R\ge\alpha$. Thus $c_i$ covers every point of $S\setminus\{p_i\}$ at radius $R$, while $m$ covers both anchors at that radius. Moreover,
\[
    \|p_i-m\|_2 \ge10R-\frac R2 =\frac{19R}{2} >R+2\alpha,
\]
where the strict inequality follows from $R\ge8\alpha$. We have therefore established
\begin{equation}
    \cost_{P\setminus\{p_i\}}(C_i)\le R, \qquad d_E(p_i,C_i)=R+2\alpha.
    \label{eq:sep-euclidean-witness}
\end{equation}

We now show that every proper nonempty subset $Q\subsetneq P$ fails to be an additive $\alpha$ weak coreset. It suffices to exhibit an exact optimal continuous $2$-center solution for $Q$ whose cost on $P$ exceeds $\OPT_2(P)+\alpha=R+\alpha$.

Suppose first that at least one anchor is omitted, and choose $z\in\{a,b\}\setminus Q$. Cover the retained polygon vertices, if any, by the origin, and use the remaining anchor itself as a center if it is retained. This uses at most two centers and covers $Q$ at radius at most $R/2$, so $r_Q:=\OPT_2(Q)\le R/2$. Take an exact optimum for $Q$, assign each point of $Q$ to one of its nearest centers, breaking ties arbitrarily, and discard every center assigned no point.
Denote the remaining center set by $\widehat C$. It is nonempty because $Q$ is nonempty, uses at most two centers,
and still covers $Q$ at radius $r_Q$. It therefore remains an exact optimum of $Q$. For each $\widehat c\in\widehat C$, choose an assigned point $q\in Q$ with $\|q-\widehat c\|_2\le r_Q$. The omitted anchor $z$ is at distance at least $2R$ from every point of $Q$, so
\[
    \|z-\widehat c\|_2 \ge\|z-q\|_2-\|q-\widehat c\|_2 \ge2R-r_Q \ge\frac{3R}{2} >R+\alpha.
\]
The last inequality follows from $R\ge8\alpha$ and $\alpha>0$. Hence
\[
    \cost_P(\widehat C) \ge d_E(z,\widehat C) >R+\alpha,
\]
which violates the required guarantee.

Next, suppose that both anchors are retained but no polygon vertex is retained. Then $Q=\{a,b\}$, and the center set $C:=\{a,b\}$ is an exact optimum of cost zero for $Q$. Every point of $S$ has distance at least $17R/2$ from $C$, so
\[
    \cost_P(C)\ge\frac{17R}{2}>R+\alpha.
\]
Thus this subset also fails.

It remains to consider a subset $Q$ retaining both anchors and at least one polygon vertex. Since $Q\ne P$, some $p_i$ is omitted. Choose any retained polygon vertex $q\in Q\cap S$. The three points $a,b,q$ have pairwise distances at least $2R$, so no two balls of radius strictly less than $R$ can cover them. Therefore $\OPT_2(Q)\ge R$. The reverse inequality follows from $Q\subseteq P$, giving
\[
    \OPT_2(Q)=R.
\]
By~\eqref{eq:sep-euclidean-witness}, the center set $C_i$ covers $Q$ at radius $R$ and is therefore an exact optimum of $Q$.
However,
\[
    \cost_P(C_i) \ge d_E(p_i,C_i) =R+2\alpha >R+\alpha =\OPT_2(P)+\alpha.
\]
This again violates the additive weak coreset guarantee.

These three cases cover every proper nonempty subset of $P$. In each case, we have exhibited at least one exact optimum of $Q$ that fails the required guarantee on $P$. Since the guarantee must hold for every exact optimum, every nonempty additive $\alpha$ weak coreset equals $P$ and has size $N+2$. For fixed $\alpha$, the integer $N$ is arbitrary. Thus no size bound depending only on $k$, $D$, and $\alpha$ is possible, even for $k=2$, $D=2$, and inputs satisfying \eqref{eq:sep-condition}. This proves Theorem~\ref{thm:sep-euclidean}.

Finally, we obtain the quantitative bound for a prescribed optimum radius $R$. Fix $\alpha>0$ and $R>0$ with $R/\alpha\ge18$, and choose 
\[
    N:=\left\lfloor\sqrt{\frac{R}{2\alpha}}\right\rfloor\ge3.
\]
Since $\pi/N\le\pi/2$, the inequality $\sin x\ge2x/\pi$ for $0\le x\le\pi/2$ gives
\[
    \delta_N =2\sin^2(\pi/N) \ge\frac8{N^2}.
\]
Consequently,
\[
    \frac{16\alpha}{\delta_N} \le2\alpha N^2 \le R,
\]
so the construction applies with this prescribed $R$. The resulting input satisfies $\OPT_2(P)=R$ and \eqref{eq:sep-condition}, and every nonempty additive $\alpha$ weak coreset has size
\[
    |Q|=|P|=N+2=\Omega\!\left(\sqrt{R/\alpha}\right).
\]
This establishes the quantitative lower bound.


\section{Transferring Planar Coreset Bounds to the Hyperbolic Plane}\label{app:planar-transfer}
We prove the existence transfer used in the discussion of related work. For a finite connected graph $G$ with positive edge lengths, let $d_G$ be its shortest-path metric and write
\[
    \cost_A^G(C):=\max_{p\in A}\min_{c\in C}d_G(p,c)
\]
for nonempty sets $A,C\subseteq V(G)$.

\begin{lemma}\label{lem:planar-transfer}
Fix integers $k,M\ge1$ and $0<\eps<1$. Suppose that every finite connected planar graph $G$ with positive edge lengths has the following property: for every nonempty client set $P_G\subseteq V(G)$, there is a nonempty subset $Q_G\subseteq P_G$ with $|Q_G|\le M$ such that
\[
    \cost_{P_G}^G(C)\le(1+\eps)\cost_{Q_G}^G(C)
\]
for every $C\subseteq V(G)$ with $1\le|C|\le k$. Then every finite nonempty set $P\subseteq\mathbb H^2$ has a nonempty subset $Q\subseteq P$ with $|Q|\le M$ such that
\[
    \cost_Q(C)\le\cost_P(C)\le(1+\eps)\cost_Q(C)
\]
for every $C\subseteq\mathbb H^2$ with $1\le|C|\le k$.
\end{lemma}

\begin{proof}
If $|P|\le M$, take $Q=P$. Otherwise, suppose that no such subset exists, and let
\[
    \mathcal Q:=\{Q\subseteq P:1\le|Q|\le M\}.
\]
For each $Q\in\mathcal Q$, choose a center set $C_Q\subseteq\mathbb H^2$ with $1\le|C_Q|\le k$ such that
\begin{equation}
    \cost_P(C_Q)>(1+\eps)\cost_Q(C_Q).
    \label{eq:planar-transfer-failure}
\end{equation}
The family $\mathcal Q$ is finite, so the terminal set
\[
    T:=P\cup\bigcup_{Q\in\mathcal Q}C_Q
\]
is finite as well.

Draw the hyperbolic geodesic segment between every pair of distinct terminals in $T$. Subdivide their union at every terminal, at intersections of segments with distinct supporting geodesics, and at the endpoints of overlapping portions. Identify coincident vertices and merge shared subsegments. Distinct supporting hyperbolic geodesics meet at most once, so only finitely many subdivisions are needed. The resulting drawing is a finite connected planar graph $G$ whose vertex set contains $T$. Its edges are the positive-length geodesic subsegments between consecutive vertices, and each edge is assigned its hyperbolic length.

For all $x,y\in T$, 
\begin{equation}
    d_G(x,y)=d_H(x,y).
    \label{eq:planar-transfer-terminal-distances}
\end{equation}
Indeed, the drawn geodesic segment from $x$ to $y$ is a graph path of length $d_H(x,y)$, giving $d_G(x,y)\le d_H(x,y)$. Conversely, every graph path is a piecewise geodesic curve in $\mathbb H^2$ with the same length as the graph path, so its length is at least $d_H(x,y)$.

Apply the assumed graph coreset property with client set $P_G:=P$ and all vertices of $G$ allowed as centers. The additional vertices are not added to the client set. Let $\widehat Q\subseteq P$ be the resulting nonempty subset of size at most $M$. Since $\widehat Q\in\mathcal Q$, the chosen set $C_{\widehat Q}$ is contained in $T\subseteq V(G)$. By~\eqref{eq:planar-transfer-terminal-distances},
\[
    \cost_P(C_{\widehat Q}) =\cost_P^G(C_{\widehat Q}) \le(1+\eps)\cost_{\widehat Q}^G(C_{\widehat Q}) =(1+\eps)\cost_{\widehat Q}(C_{\widehat Q}),
\]
contradicting~\eqref{eq:planar-transfer-failure} for $Q=\widehat Q$. Thus a subset satisfying the upper inequality exists. The lower inequality follows from $Q\subseteq P$.
\end{proof}

Bourneuf and Pilipczuk~\cite[Theorem~6.1]{BourneufPilipczuk2025} provide the required graph guarantee directly in the
$(1+\eps)$ convention. In the planar case, their size bound can be written as
\[
    (k+1)^{\eps^{-O(1)}}.
\]
Taking all graph vertices as potential centers and applying Lemma~\ref{lem:planar-transfer} therefore gives the same strong coreset existence bound in $\mathbb H^2$.

Kluk et al.~\cite[Definition~1.9 and Theorem~1.10]{KlukEtAl2026} use the convention
\[
    \cost_{Q_G}^G(C)\ge(1-\eta)\cost_{P_G}^G(C)
\]
and give a planar coreset-size bound of $(1/\eta)^{2^{O(k)}}$. Set
\[
    \eta:=\frac{\eps}{1+\eps}.
\]
Then $(1-\eta)^{-1}=1+\eps$, so their guarantee has the form required by Lemma~\ref{lem:planar-transfer}. Since $\eta^{-1}=(1+\eps)/\eps\le2/\eps$, this gives the complementary strong coreset existence bound
\[
    (2/\eps)^{2^{O(k)}}
\]
in $\mathbb H^2$.

This transfer preserves the coreset cardinality bounds without introducing additional approximation error. The auxiliary graph, however, depends on center sets witnessing the failure of every candidate subset under the contradictory assumption. The proof does not supply an efficient method to obtain these center sets or construct the auxiliary graph from $P$. It therefore establishes existence bounds but does not transfer the running times of the graph algorithms.

\end{document}